\documentclass[12pt,letterpaper]{article}
\usepackage{amsmath,amsfonts,amsthm,amssymb}

\theoremstyle{plain}

\numberwithin{equation}{section}
\newtheorem{thm}{Theorem}[section]
\newtheorem{lem}[thm]{Lemma}

\newenvironment{exam}[1]{\medskip%
  \setlength{\rightmargin}{\leftmargin}%
  {\noindent\textbf{Example #1}.\enspace}%
  }%

\newcounter{cond}

\newcommand{\real}{{\mathbb R}}
\newcommand{\integers}{{\mathbb Z}}
\newcommand{\ascript}{{\mathcal A}}

\newcommand{\nuhat}{\widehat{\nu}}
\newcommand{\phihat}{\widehat{\phi}}
\newcommand{\muhat}{\widehat{\mu}}
\newcommand{\fhat}{\widehat{f}}

\newcommand{\cupdot}{\mathbin{\cup{\hskip-5.4pt}^\centerdot}\,}
\newcommand{\bigcupdotionen}{{\bigcup _{i=1}^n}{\hskip-8pt}^\centerdot{\hskip 8pt}}
\newcommand{\bigcupdotitwon}{{\bigcup _{i=2}^n}{\hskip-8pt}^\centerdot{\hskip 8pt}}
\newcommand{\bigcupdotijn}{{\bigcup _{i=j}^n}{\hskip-8pt}^\centerdot{\hskip 8pt}}
\newcommand{\bigcupdotikn}{{\bigcup _{i=k}^n}{\hskip-8pt}^\centerdot{\hskip 8pt}}
\newcommand{\bigcupdotijonen}{{\bigcup _{i=j+1}^n}{\hskip-13pt}^\centerdot{\hskip 13pt}}
\newcommand{\bigcupdotimn}{{\bigcup _{i=m}^n}{\hskip-10pt}^\centerdot{\hskip 10pt}}
\newcommand{\subcupdot}{\mathbin{\cup{\hskip-4pt}^\centerdot}\,}

\newcommand{\ab}[1]{\left|#1\right|}
\newcommand{\brac}[1]{\left\{#1\right\}}
\newcommand{\paren}[1]{\left(#1\right)}
\newcommand{\sqbrac}[1]{\left[#1\right]}

\begin{document}

\title{QUANTUM INTEGRALS AND\\ANHOMOMORPHIC LOGICS}
\author{Stan Gudder\\ Department of Mathematics\\
University of Denver\\ Denver, Colorado 80208\\
sgudder@math.du.edu}
\date{}
\maketitle

\begin{abstract}
The full anhomomorphic logic of coevents $\ascript ^*$ is introduced. Atoms of $\ascript ^*$ and embeddings of the event set $\ascript$ into $\ascript ^*$ are discussed. The quantum integral over an event $A$ with respect to a coevent $\phi$ is defined and its properties are treated. Integrals with respect to various coevents are computed. Reality filters such as preclusivity and regularity of coevents are considered. A quantum measure $\mu$ that can be represented as a quantum integral with respect to a coevent $\phi$ is said to 1-generate $\phi$. This gives a stronger reality filter that may produce a unique coevent called the ``actual reality'' for a physical system. What we believe to be a more general filter is defined in terms of a double quantum integral and is called 2-generation. It is shown that ordinary measures do not 1 or 2-generate coevents except in a few simple cases. Examples are given which show that there are quantum measures that 2-generate but do not 1-generate coevents. Examples also show that there are coevents that are
2-generated but not 1-generated. For simplicity only finite systems are considered.
\end{abstract}

\section{Introduction}  
Quantum measure theory and anhomomorphic logics have been studied for the past 16 years
\cite{gt09, gud1, gud2, gud4, sal02, sor94, sor07, sor09, sw08}. The main motivations for these studies have been investigations into the histories approach to quantum mechanics and quantum gravity and cosmology. The author has recently introduced a quantum integration theory \cite{gud3} and the present article presents some connections between this theory in a slightly different setting and anhomomorphic logics. It turns out that classical logic, in which truth functions are given by homomorphisms, is not adequate for quantum mechanical studies. instead, one must employ truth functions that are not homomorphism and this is the origin of the term anhomomorphic logic. These more general truth functions are Boolean-valued functions on the set $\ascript$ of quantum events (or propositions) and are called coevents. We denote the set of coevents with reasonable properties by $\ascript ^*$ and call $\ascript ^*$ the full anhomomorphic logic. The elements of $\ascript ^*$ correspond to potential realities for a quantum system. For
$A\in\ascript$, $\phi\in\ascript ^*$, $\phi (A)=1$ if and only if the event $A$ occurs (or the proposition $A$ is true) in the reality described by $\phi$. The main goal of the theory is to find the ``actual reality'' $\phi _a$ for the system.

Even for a system with a small finite sample space $\Omega$ of cardinality $n$, the cardinality of $\ascript ^*$ is
$2^{(2^n-1)}$ which is huge. Thus, finding $\phi _a$ in this huge space may not be an easy task. It is thus important to develop filters or criteria for reducing this number of potential realities to make the selection of $\phi _a$ more manageable. One of the main filters used in the past has been preclusivity. Nature has provided us with an underlying quantum measure $\mu$ on $\ascript$. The measure $\mu (A)$ is sometimes interpreted as the propensity of occurrence for the event $A$. If $\mu (A)=0$, then $A$ does not occur and we say that $A$ is \textit{precluded}. If
$\phi (A)=0$ for every precluded event $A$, then $A$ is \textit{preclusive}. It is postulated that $\phi _a$ must be preclusive \cite{gt09, sor94, sor07}. Unfortunately, there may still be many preclusive coevents so this criteria does not specify $\phi _a$ uniquely. In this article we propose a stronger filter that may uniquely determine $\phi _a$.

Motivated by previous work on quantum integration \cite{gud3}, we define an integral $\int _Afd\phi$ over
$A\in\ascript$ of a real-valued function on $\Omega$ with respect to a coevent $\phi$. If $\mu$ is a quantum measure on $\ascript$ and if $\mu (A)=\int _Afd\phi$ for all $A\in\ascript$ where $f$ is a strictly positive function, we say that
$\mu$ 1-\textit{generates} $\phi$. If $\mu$ 1-generates $\phi$, then $\phi$ is automatically preclusive (relative to
$\mu$). We do not know whether a 1-generated coevent $\phi$ is unique but we have a partial result in that direction. In any case, 1-generated coevents provide a much stronger filter than preclusivity. Unfortunately, as we discuss in more detail later, one cannot expect an arbitrary quantum measure $\mu$ to 1-generate a coevent and we shall show that this only holds for a very restricted set of quantum measures. For this reason we introduce what we believer is a more general method that holds for a much larger set of (but not all) quantum measures. We say that $\mu$
2-\textit{generates} $\phi$ if
\begin{equation*}
\mu (A)=\int _A\sqbrac{\int _Af(\omega ,\omega ')d\phi (\omega )}d\phi (\omega ')
\end{equation*}
for all $A\in\ascript$ where $f$ is a symmetric, strictly positive function on $\Omega\times\Omega$. Again, if $\mu$
2-generates $\phi$, then $\phi$ is preclusive (relative to $\mu$). A result which we find interesting is that except for a few simple cases, no ordinary measure 1- or 2-generates a coevent. Thus, the concept of generating coevents is essentially purely quantum mechanical. This article includes many examples that illustrate various concepts. For simplicity, we only consider finite sample spaces.

We now briefly summarize the contents of the paper. In  Section~2, we define the full anhomomorphic logic
$\ascript ^*$. We briefly discuss the additive, multiplicative and quadratic sublogics of $\ascript ^*$. These sublogics have been considered in the past and it has not yet been settled which is the most suitable or whether some other sublogic is preferable. For this reason and for generality, we do not commit to a particular sublogic here. We point out that $\ascript ^*$ is a Boolean algebra and we discuss the atoms of $\ascript ^*$. Two embeddings of $\ascript$ into $\ascript ^*$ denoted by $A\mapsto A_*$ and $A\mapsto A^*$ are treated.

Section~3 introduces the quantum integral $\int fd\phi$ with respect to the coevent $\phi$. Properties of this integral and the more general integral $\int _Afd\phi$ are discussed. Integrals with respect to various coevents such as $A_*$ and $A^*$ are computed. Reality filters are considered in Section~4. Preclusive and regular coevents are discussed. Most of the section is devoted to the study of 1- and 2-generated coevents. Section~5 presents some general theorems. A uniqueness result shows that if $\phi ,\psi\in\ascript ^*$ are regular and $\mu$ 1-generates both $\phi$ and $\psi$, then $\phi =\psi$. Expansions of quantum measures and coevents are defined. It is shown that $\mu$
1- or 2-generates $\phi$ if and only if expansions of $\mu$ 1- or 2-generate corresponding expansions of $\phi$. It is shown that if $\phi$ is 1-generated by $\mu$ and $\phi (A)\ne 0$ whenever $\mu (A)=0$, then $\phi$ is 2-generated by $\mu$. Sections~6 and ~7 are devoted to examples of 1- and 2-generated coevents. For instance, it is shown that there are quantum measures that 2-generate coevents but do not 1-generate coevents. Examples of coevents that are 2-generated but not 1-generated are given. Also, examples of coevents that are not 1- or 2-generated are presented.

\section{Full Anhomomorphic Logics} 
Let $\Omega$ be a finite nonempty set with cardinality $\ab{\Omega}<\infty$. We call $\Omega$ a
\textit{sample space}. The elements of $\Omega$ correspond to outcomes or trajectories of an experiment or physical system and the collection of subsets $2^\Omega$ of $\Omega$ correspond to possible events. We can also think of the sets in $2^\Omega$ as propositions concerning the system. Contact with reality is given by a truth function
$\phi\colon 2^\Omega\to\brac{0,1}$. The function $\phi$ specifies what actually happens where we interpret
$\phi (A)=1$ to mean that $A$ is true or occurs and $\phi (A)=0$ means that $A$ is false or does not occur. It is convenient to view $\brac{0,1}$ as the two element Boolean algebra $\integers _2$ with the usual multiplication and addition given by $0\oplus 0=1\oplus 1=0$ and $0\oplus 1=1\oplus 0=1$.

For $\omega\in\Omega$ we define the \textit{evaluation map} $\omega ^*\colon 2^\Omega\to\integers _2$ by
\begin{equation*}
\omega (A)=\begin{cases}1&\text{if }\omega\in A\\0&\text{if }\omega\not\in A\end{cases}
\end{equation*}
For classical systems, it is assumed that a truth function $\phi$ is a homomorphism; that is, $\phi$ satisfies:
\begin{list} {(H\arabic{cond})}{\usecounter{cond}
\setlength{\rightmargin}{\leftmargin}}
\item $\phi (\Omega )=1$\ (unital)
\item $\phi (A\cupdot B)=\phi (A)\oplus\phi (B)$ whenever $A\cap B=\emptyset$\ (additive)
\item $\phi (A\cap B)=\phi (A)\phi (B)$\ (multiplicative)
\end{list}
In (H2) $A\cupdot B$ denotes $A\cup B$ whenever $A\cap B=\emptyset$. It is well-known that $\phi$ is a homomorphism if and only if $\phi =\omega ^*$ for some $\omega\in\Omega$. Thus, there are $\ab{\Omega}$ truth functions for classical systems.

As discussed in \cite{sor94, sor07, sor09}, for a quantum system a truth function need not be a homomorphism. However, a quantum truth function should satisfy some requirements or else there would be no theory at all. Various proposals have been presented concerning what these requirements should be \cite{gt09, sor94, sor07, sor09}. In \cite{sor94} it is assumed that quantum truth functions satisfy (H2) and these are called \textit{additive} truth functions, while in \cite{gt09, sor09} it is assumed that quantum truth functions satisfy (H3) and these are called
\textit{multiplicative} truth functions. In \cite{gud4} it is argued that quantum truth functions need not satisfy (H1), (H2) or (H3) but should be \textit{quadratic} or \textit{grade}-2 \textit{additive} in the sense that

\medskip
\noindent (H4)\enspace 
$\phi (A\cupdot B\cupdot C)=\phi (A\cupdot B)\oplus\phi (A\cupdot C)\oplus\phi (B\cupdot C)
\oplus\phi (A)\oplus\phi (B)\oplus\phi (C)$\bigskip

If $\phi ,\psi\colon 2^\Omega\!\to\integers _2$ are truth functions, we define $\phi\psi$ by
$(\phi\psi )(A)\!=\!\phi (A)\psi (A)$ and $\phi\oplus\psi$ by $(\phi\oplus\psi )(A)=\phi (A)\oplus\psi (A)$ for all
$A\in 2^\Omega$. We make the standing assumption that $\phi (\emptyset )=0$ for any truth function and we do admit the constant $0$ function as a truth function and denote it by $0$. We also define $1$ as the truth function given by $1(A)=1$ for all $A\ne\emptyset$. It can be shown \cite{gt09, gud4, sor09} that $\phi$ is additive if and only if $\phi$ is $0$ or a degree-1 polynomial
\begin{equation*}
\phi =\omega _1^*\oplus\cdots\oplus\omega _n^*
\end{equation*}
and that $\phi$ is multiplicative if and only if $\phi$ is a monomial
\begin{equation*}
\phi =\omega _1^*\omega _2^*\cdots\omega _n^*
\end{equation*}
Moreover, one can show \cite{gt09, gud4} that $\phi$ is quadratic if and only if $\phi$ is a degree-1 polynomial or
$\phi$ is a degree-2 polynomial of the form
\begin{equation*}
\phi =\omega _1^*\oplus\cdots\oplus\omega _n^*\oplus\omega _i^*\omega _j^*
  \oplus\cdots\oplus\omega _n^*\omega _s^*
\end{equation*}
Notice that we do not allow a constant term in our polynomials.

Since it is not completely clear what conditions a quantum truth function should satisfy and for the sake of generality, we shall not commit to any particular type of truth function here. We use the notation
$\Omega _n=\brac{\omega _1,\ldots ,\omega _n}$, $\ascript _n=2^{\Omega _n}$ for an $n$-element sample space. Leaving out the truth functions satisfying $\phi (\emptyset )=1$ we have $2^{(2^n-1)}$ admissible truth functions. We denote this set of truth functions by $\ascript _n^*$ or $\ascript ^*$ when no confusion arises. Similarly, we sometimes denote the set of events $\ascript _n$ by $\ascript$. We call $\ascript ^*$ the \textit{full anhomomorphic logic} and the elements of $\ascript ^*$ are called \textit{coevents}. In $\ascript ^*$ there are $n$ classical, $2^n$ additive, $2^n$ multiplicative and $2^{n(n+1)/2}$ quadratic coevents. All nonzero multiplicative coevents are unital and half of the additive and quadratic coevents are unital, namely those with an odd number of summands. It can be shown that any nonzero coevent $\phi$ can be uniquely written as a polynomial in the evaluation maps. We call this the
\textit{evaluation map representation} of $\phi$. Notice that $\ascript _1^*$ consists of the $0$ coevent and the single classical coevent $\omega _1^*$. The following examples discuss $\ascript _2^*$ and $\ascript _3^*$.

\begin{exam}{1}    
The full anhomomorphic logic $\ascript _2^*$ contain $2^3=8$ elements. We list them according to types. The $0$ coevent is type $(0)$ and the two classical coevents $\omega _1^*$, $\omega _2^*$ are type $(1)$. The additive coevent $\omega _1^*\oplus\omega _2^*$ is type $(1,2)$. The multiplicative coevent $\omega _1^*\omega _2^*$ is type $(12)$. The two quadratic coevents $\omega _1^*\oplus\omega _1^*\omega _2^*$,
$\omega _2^*\oplus\omega _1^*\omega _2^*$ are type $(1,12)$. Finally, the quadratic coevent
$1=\omega _1^*\oplus\omega _2^*\oplus\omega _1^*\omega _2$ is type $(1,2,12)$.\hfill\qedsymbol
\end{exam} 

\begin{exam}{2}    
The full anhomomorphic logic $\ascript _3^*$ contains $2^7=128$ elements. There are too many to list so we shall just consider some of them. The coevent $\omega _1^*$ is one of the three classical type (1) coevents and
$\omega _1^*\oplus\omega _2^*$ is one of the three type $(1,2)$ coevents. The coevent $\omega _1^*\omega _2^*$ is one of the three type $(12)$ coevents and $\omega _1^*\oplus\omega _2^*\oplus\omega _2^*$ is the only type $(1,2,3)$ coevent. There is one type $(123)$ coevent $\omega _1^*\omega _2^*\omega _3^*$ and there are six type $(1,12)$ coevents, one being $\omega _1^*\oplus\omega _1^*\omega _2^*$. There are three type $(1,23)$ coevents, one being $\omega _1^*\oplus\omega _2^*\omega _3^*$ and three type $(1,2,12)$ coevents, one being
$\omega _1^*\oplus\omega _2^*\oplus\omega _1^*\omega _2^*$. The coevent
$\omega _1^*\oplus\omega _2^*\oplus\omega _1^*\omega _3^*$ is one of the six type $(1,2,13)$ coevents and
$\omega _1^*\oplus\omega _2^*\oplus\omega _3^*\oplus\omega _1^*\omega _2^*$ is one of the three type
$(1,2,3,12)$ coevents. Our last example is the type $(1,2,3,12,13,23,123)$ coevent\medskip

$\hfill 1=\omega _1^*\oplus\omega _2^*\oplus\omega _3^*\oplus\omega _1^*\omega _2^*
  \oplus\omega _1^*\omega _3^*\oplus\omega _2^*\omega _3^*\oplus\omega _1^*\omega _2^*\omega _3^*
  \hfill\qedsymbol$
\end{exam}
\medskip

It is of interest to note that the coevent $1$ in $\ascript _n^*$ has the form
\begin{equation}         
\label{eq21}
1=\bigoplus _{i=1}^n\omega _i^*\oplus\bigoplus _{i<j=1}^n\omega _i^*\omega _j^*\oplus\cdots\oplus
  \omega _1^*\omega _2^*\cdots\omega _n^*
\end{equation}
To show this, let $A\in\ascript _n$ and assume without loss of generality that $A=\brac{\omega _1,\ldots ,\omega _m}$. Denoting the right side of \eqref{eq21} by $\phi$, the number of terms of $\phi$ that are $1$ on $A$ becomes
\begin{equation*}
\binom{m}{1}+\binom{m}{2}+\cdots +\binom{m}{m}=2^m-1
\end{equation*}
Since $2^m-1$ is odd, we conclude that $\phi (A)=1$.

For $\phi ,\psi\in\ascript ^*$ we define $\phi\le\psi$ if $\phi (A)\le\psi (A)$ for all $A\in\ascript$. It is clear that $\le$ is a partial order on $\ascript ^*$ and that $0\le\phi\le 1$ for all $\phi\in\ascript ^*$. For $\phi\in\ascript ^*$ we define the \textit{complement} $\phi '$ of $\phi$ by $\phi '=1\oplus\phi$. It is easy to check that $(\ascript ^*,0,1,\le ,{}')$ is a Boolean algebra in which the meet and join are given by
\begin{align*}
(\phi\wedge\psi )(A)&=\min\paren{\phi (A),\psi (A)}=\phi (A)\psi (A)\\
   (\phi\vee\psi )(A)&=\max\paren{\phi (A),\psi (A)}=\phi (A)\oplus\psi (A)\oplus\phi (A)\psi (A)
\end{align*}
Examples of meets and joins are:
\begin{align*}
&\omega _1^*\wedge\omega _2^*\wedge\cdots\wedge\omega _m^*=\omega _1^*\omega _2^*\cdots\omega _m^*\\
&\omega _1^*\vee\omega _2^*=\omega _1^*\oplus\omega _2^*\oplus\omega _1^*\omega _2^*\\
&(\omega _1^*\oplus\omega _1^*\omega _2^*)\vee (\omega _2^*\oplus\omega _1^*\omega _2^*)
  =\omega _1^*\oplus\omega _2^*\\
&(\omega _1^*\oplus\omega _1^*\omega _2^*\oplus\omega _1^*\omega _3^*\oplus
  \omega _1^*\omega _2^*\omega _3^*)
  \vee (\omega _2^*\oplus\omega _1^*\omega _2^*\oplus\omega _2^*\omega _3^*
  \oplus\omega _1^*\omega _2^*\omega _3^*)\\
&\quad =\omega _1^*\oplus\omega _2^*\oplus\omega _1^*\omega _3^*\oplus\omega _2^*\omega _3^*
\end{align*}
The \textit{atoms} of $\ascript ^*$ are the minimal nonzero elements of $\ascript ^*$ and every nonzero coevent $\phi$ is the unique join of the atoms below $\phi$. Thus, every atom has the form $\phi _A$ for $\emptyset\ne A\in\ascript$ where $\phi _A(B)=1$ if and only if $B=A$. The atoms of $\ascript _2^*$ are
$\phi _{\brac{\omega _1}}=\omega _1^*\oplus\omega _1^*\omega _2^*$,
$\phi _{\brac{\omega _2}}=\omega _2^*\oplus\omega _1^*\omega _2^*$, and
$\phi _\Omega =\omega _1^*\omega _2^*$. We then have
$\omega _1^*=\phi _{\brac{\omega _1}}\vee\phi _\Omega$,
$\omega _2^*=\phi _{\brac{\omega _2}}\vee\phi _\Omega$,
$\omega _1^*\oplus\omega _2^*=\phi _{\brac{\omega _1}}\vee\phi _{\brac{\omega _2}}$ and
$1=\phi _{\brac{\omega _1}}\vee\phi _{\brac{\omega _2}}\vee\phi _\Omega$. The seven atoms of $\ascript _3^*$ are the type $(123)$, the three type $(1,12,13,123)$ and the three type $(12,123)$ coevents.

We now describe the atoms of $\ascript _n^*$ in terms of their evaluation map representations. Without loss of generality suppose $A\in\ascript _n$ has the form $A=\brac{\omega _1,\ldots ,\omega _m}$ and let
$\phi=\omega _1^*\omega _2^*\cdots\omega _m^*$. We claim that
\begin{equation}         
\label{eq22}
\phi _A=\phi\oplus\bigoplus _{i=m+1}^n\phi\omega _i^*\oplus\bigoplus _{i<j=m+1}^n\phi\omega _i^*\omega _j^*
  \oplus\cdots\oplus\phi\omega _{m+1}^*\cdots\phi _n^*
\end{equation}
To show this, let $\psi$ be the right side of \eqref{eq22}. It is clear that $\psi (A)=\phi (A)=1$. Now suppose that
$B\in\ascript _n$ with $\emptyset\ne B\ne A$. If $A\not\subseteq B$, then $\psi (B)=0$. If $A\subseteq B$, then we can assume without loss of generality that
$B=\brac{\omega _1,\ldots ,\omega _m,\omega _{m+1},\ldots ,\omega _{m+r}}$. The number of terms on the right side of \eqref{eq22} that are $1$ on $B$ is
\begin{equation*}
1+\binom{r}{1}+\binom{r}{2}+\cdots +\binom{r}{r}=2^r
\end{equation*}
Since $2^r$ is even, we conclude that $\psi (B)=0$ and this proves \eqref{eq22}.

We next give a natural embedding of the Boolean algebra $\ascript$ into $\ascript ^*$. For $A\in\ascript$, we define
$A_*\in\ascript$ by $A_*(B)=1$ if and only if $\emptyset\ne B\subseteq A$. Notice that $\emptyset _*=0$ and
$\Omega _*=1$. Also, $\brac{\omega }_*=\phi _{\brac{\omega}}$ so $\brac{\omega}_*$ has the form \eqref{eq22} with $m=1$ and $\phi =\omega ^*=\omega _1^*$. The evaluation map representation of $A_*$ is a generalization of \eqref{eq22} that we shall discuss later. We denote the complement of a set $A$ by $A'$.

\begin{thm}       
\label{thm21}
{\rm (a)}\enspace $A\subseteq B$ if and only if $A_*\le B_*$.
{\rm (b)}\enspace $(A\cap B)_*=A_*\wedge B_*$.
\end{thm}
\begin{proof}
(a)\enspace Suppose that $A\subseteq B$. If $A_*(C)=1$, then $\emptyset\ne C\subseteq A$. Hence, $C\subseteq B$ so $B_*(C)=1$. It follows that $A_*\le B_*$. Conversely, suppose that $A_*\le B_*$. If $A\not\subseteq B$, then
$A\cap B'\ne\emptyset$ and $A\cap B'\subseteq A$. Hence, $A_*(A\cap B')=1$ so $B_*(A\cap B')=1$. We conclude that $A\cap B'\subseteq B$ which is a contradiction. Therefore, $A\subseteq B$.\enspace
(b)\enspace If $A_*\wedge B_*(C)=1$, then $\min\paren{A_*(C),B_*(C)}=1$. Hence, $A_*(C)=B_*(C)=1$ so
$C\subseteq A\cap B$. It follows that $(A\cap B)_*(C)=1$. Conversely, if $(A\cap B)_*(C)=1$, then
$C\subseteq A\cap B$. Hence, $A_*(C)=B_*(C)=1$ so $(A_*\wedge B_*)(C)=1$. Therefore,
$(A\cap B)_*=A_*\wedge B_*$.
\end{proof}

In general, $A\mapsto A_*$ is not a Boolean homomorphism. One reason is that $(A')_*\ne (A_*)'$. Indeed, if
$B\cap A\ne\emptyset$ and $B\cap A'\ne\emptyset$, then $(A')_*(B)=A_*(B)=0$. Another reason is that
$(A\cup B)_*\ne A_*\vee B_*$. To see this, suppose that $C\not\subseteq A$, $C\not\subseteq B$ and
$C\subseteq A\cup B$. Then $(A\cup B)_*(C)=1$ while $A_*(C)=B_*(C)=0$. To describe the evaluation map representation of $A_*$ we can assume without loss of generality that $A\subseteq\Omega _n$ and
$A=\brac{\omega _1,\ldots ,\omega _m}$, $m\le n$. We claim that $A_*$ is the sum of all monomials that contain at least one of the elements $\omega _i^*$, $i=1,\ldots ,m$ as a factor. For instance, in $\Omega _3$ we have
\begin{equation*}
\brac{\omega _1,\omega _2}_*=\omega _1^*\oplus\omega _2^*\oplus\omega _1^*\omega _2^*\oplus
  \omega _1^*\omega _3^*\oplus\omega _2^*\omega _3^*\oplus\omega _1^*\omega _2^*\omega _3^*
\end{equation*}
as is easy to check. Instead of giving the general proof of this claim, which is tedious, we present some other examples.

\begin{exam}{3}    
In $\Omega _4$ with $A=\brac{\omega _1,\omega _2}$ our claim states that
\begin{align}         
\label{eq23}
A_*&=\omega _1^*\oplus\omega _2^*\oplus\omega _1^*\omega _2^*\oplus\omega _1^*\omega _3^*
  \oplus\omega _1^*\omega _4^*\oplus\omega _2^*\omega _3^*\oplus\omega _2^*\omega _4^*\\
  &\quad\oplus\omega _1^*\omega _2^*\omega _3^*\oplus\omega _1^*\omega _2^*\omega _4^*
  \oplus\omega _1^*\omega _3^*\omega _4^*\oplus\omega _2^*\omega _3^*\omega _4^*
  \oplus\omega _1^*\omega _2^*\omega _3^*\omega _4^*\notag
\end{align}
Let $\psi$ be the right side of \eqref{eq23}. Clearly $\psi (B)=1$ for every $\emptyset\ne B\subseteq A$ and
$\psi (B)=0$ for every $B\subseteq A'$. Next we have
\begin{equation*}
\psi\paren{\brac{\omega _1,\omega _3}}=\psi\paren{\brac{\omega _1,\omega _4}}
  =\psi\paren{\brac{\omega _2,\omega _3}}=\psi\paren{\brac{\omega _2,\omega _4}}=0
\end{equation*}
Also, $\psi (\Omega _4)=0$ and
\begin{equation*}
\psi\paren{\brac{\omega _1,\omega _2,\omega _3}}=\psi\paren{\brac{\omega _1,\omega _2,\omega _4}}=0
\end{equation*}
It follows that $\psi =A_*$. As another example, let $B=\brac{\omega _1,\omega _2,\omega _3}$. Our claim states that
\begin{equation*}
B_*=A_*\oplus\omega _3^*\oplus\omega _3^*\omega _4^*
\end{equation*}
and this is easily verified.\hfill\qedsymbol
\end{exam}
\medskip

We now present another embedding of $\ascript$ into $\ascript ^*$. For $A=\brac{\omega _1,\ldots ,\omega _m}$, let $A^*$ be the sum of all monomials that contain only elements $\omega _1^*,\ldots ,\omega _n^*$ as factors. Thus,
\begin{equation}         
\label{eq24}
A^*=\bigoplus _{i=1}^m\omega _i^*\oplus\bigoplus _{i<j=1}^m\omega _i^*\omega _j^*\oplus\cdots\oplus
  \omega _1^*\omega _2^*\cdots\omega _m^*
\end{equation}
Notice that $\emptyset ^*=0$, $\Omega ^*=1$ and $\brac{\omega}^*=\omega ^*$.

\begin{thm}       
\label{thm22}
{\rm (a)}\enspace $A^*(B)=1$ if and only if $B\cap A\ne\emptyset$.
{\rm (b)}\enspace $A\subseteq B$ if and only if $A^*\le B^*$.
{\rm (c)}\enspace $(A\cup B)^*=A^*\vee B^*$.
\end{thm}
\begin{proof}
(a)\enspace If $B\cap A=\emptyset$, then clearly $A^*(B)=0$. Suppose $B\cap A\ne\emptyset$ and $\ab{B\cap A}=r$. Representing $A^*$ as in \eqref{eq24} we have
\begin{equation}         
\label{eq25}
A^*(B)=\bigoplus _{i=1}^m\omega _i^*(B)\oplus\bigoplus _{i<j=1}^m\omega _i^*\omega _j^*(B)\oplus\cdots\oplus
  \omega _1^*\omega _2^*\cdots\omega _m^*(B)
\end{equation}
The number of terms on the right side of \eqref{eq25} that equal $1$ are
\begin{equation*}
\binom{r}{1}+\binom{r}{2}+\cdots +\binom{r}{r}=2^r-1
\end{equation*}
Since $2^r-1$ is odd, $A^*(B)=1$.
(b)\enspace Suppose that $A\subseteq B$. If $A^*(C)=1$, then by (a), $C\cap A\ne\emptyset$. Hence,
$C\cap B\ne\emptyset$ so by (a) $B^*(C)=1$. Therefore, $A^*\le B^*$. Conversely, suppose that $A^*\le B^*$. If
$A\not\subseteq B$, then $A\cap B'\ne\emptyset$. Since $A\cap B'\subseteq A$, by (a) $A^*(A\cap B')=1$. Hence, $B^*(A\cap B')=1$ so by (a) $(A\cap B')\cap B\ne\emptyset$ which is a contradiction. Hence, $A\cap B'=\emptyset$ so $A\subseteq B$.
(c)\enspace If $(A\cup B)^*(C)=1$, then $C\cap (A\cup B)\ne\emptyset$. Hence, $C\cap A\ne\emptyset$ or
$C\cap B\ne\emptyset$ so $A^*(C)=1$ or $B^*(C)=1$. Hence,
\begin{equation*}
(A^*\vee B^*)(C)=\max\paren{A^*(C),B^*(C)}=1
\end{equation*}
Conversely, if $(A^*\vee B^*)(C)=1$ then $A^*(C)=1$ or $B^*(C)=1$. Hence, $C\cap A\ne\emptyset$ or
$C\cap B\ne\emptyset$. It follows that $C\cap (A\cup B)\ne\emptyset$ so $(A\cup B)^*(C)=1$. We conclude that
$(A\cup B)^*=A^*\vee B^*$.
\end{proof}

As before $A\mapsto A^*$ is not a Boolean homomorphism because $(A')^*\ne (A^*)'$ and
$(A\cap B)^*\ne A^*\wedge B^*$ in general. For the first case, let $B$ satisfy $B\cap A\ne\emptyset$ and
$B\cap A'\ne\emptyset$. Then
\begin{equation*}
A^*(B)=(A')^*(B)=1
\end{equation*}
Hence, $(A')^*\ne (A^*)'$. For the second case, let $C$ satisfy $C\cap A$, $C\cap B\ne\emptyset$ but
$C\cap (A\cap B)=\emptyset$. Then $(A\cap B)^*(C)=0$ but
\begin{equation*}
(A^*\wedge B^*)(C)=\min\paren{A^*(C),B^*(C)}=1
\end{equation*}
Hence, $(A\cap B)^*\ne A^*\wedge B^*$. It is interesting to note that $A_*\le A^*$.

We close this section by mentioning that the coevent $\psi _A(B)=1$ if and only if $A\subseteq B$ is much simpler than $A_*$ or $A^*$. If $A=\brac{\omega _1,\ldots ,\omega _m}$, then
$\psi _A=\omega _1^*\omega _2^*\cdots\omega _m^*$ which is multiplicative.

\section{Quantum Integrals} 
As in Section~2, $\Omega$ is a finite set, $\ascript = 2^\Omega$ and $\ascript ^*$ is the full anhomomorphic logic. Following \cite{gud3}, for $f\colon\Omega\to\real$ and $\phi\in\ascript ^*$, we define the \textit{quantum integral}
($q$-\textit{integral}, for short)
\begin{equation}         
\label{eq31}
\int fd\phi =\int _0^\infty\phi\paren{\brac{\omega\colon f(\omega )>\lambda}}d\lambda
  -\int _0^\infty\phi\paren{\brac{\omega\colon f(\omega )<-\lambda}}d\lambda
\end{equation}
where $d\lambda$ denotes Lebesgue measure on $\real$. Any $f\colon\Omega\to\real$ has a unique representation $f=f_1-f_2$ where $f_1,f_2\ge 0$ and $f_1f_2=0$. It follows from \eqref{eq31} that
\begin{equation*}
\int fd\phi =\int f_1d\phi -\int f_2d\phi
\end{equation*}
We conclude that $q$-integrals are determined by $q$-integrals of nonnegative functions. In fact, all the $q$-integrals that we consider will be for nonnegative functions.

Denoting the characteristic function of a set $A$ by $\chi _A$, any nonnegative function $f\colon\Omega\to\real ^+$ has the canonical representation
\begin{equation}         
\label{eq32}
f=\sum _{i=1}^n\alpha _i\chi _{A_i}
\end{equation}
where $0<\alpha _1<\cdots <\alpha _n$ and $A_i\cap A_j=\emptyset$, $i\ne j=1,\ldots ,n$. Thus, $\alpha _i$ are the nonzero values of $f$ and $A_i=f^{-1}(\alpha _i)$, $i=1,\ldots ,n$. Since $f\ge 0$ we can write \eqref{eq31} as
\begin{equation}         
\label{eq33}
\int fd\phi =\int _0^\infty\phi\paren{\brac{\omega\colon f(\omega )>\lambda}}d\lambda
\end{equation}
and it follows from \eqref{eq32} and \eqref{eq33} that
\begin{align}         
\label{eq34}
\int\!\!fd\phi&=\alpha _1\phi\!\paren{\bigcupdotionen A_i}\!\!+\!(\alpha _2-\alpha _1)\phi\!\paren{\bigcupdotitwon a_i}
  \!\!+\cdots +\!(\alpha _n-\alpha _{n-1})\phi (A_n)\\
\label{eq35}         
  &=\sum _{j=1}^n\alpha _j\sqbrac{\phi\paren{\bigcupdotijn A_i}-\phi\paren{\bigcupdotijonen A_i}}
\end{align}
In \eqref{eq35} we use the convention that $\cupdot _{i=n+1}^nA_i=\emptyset$.

It is clear from \eqref{eq34} that $\int fd\phi\ge 0$ and $\int\alpha\chi _Ad\phi =\alpha\phi (A)$ for all $\alpha\ge 0$. In particular, $\phi (A)=\int\chi _Ad\phi$ so the $q$-integral generalizes the coevent $\phi$. Also it is easy to check that
$\int\alpha fd\phi =\alpha\int fd\phi$ and that
\begin{equation*}
\int (\alpha +f)d\phi =\int\alpha d\phi +\int fd\phi =\alpha\phi (A)+\int fd\phi
\end{equation*}
for all $\alpha\ge 0$. The $q$-integral is nonlinear, in general. For example, suppose that $A,B\in\ascript$ are disjoint nonempty sets, $0<\alpha <\beta$ and $\phi\in\ascript ^*$ satisfies $\phi (A\cupdot B)\ne\phi (A)+\phi (B)$. Then by \eqref{eq35} we have
\begin{align*}
\int\paren{\alpha\chi _A+\beta\chi _B}d\phi&=\alpha\sqbrac{\phi (A\cupdot B)-\phi (B)}+\beta\phi (B)\\
  &\ne\alpha\phi (A)+\beta\phi (B)=\int\alpha\chi _Ad\phi +\int\beta\chi _Bd\phi
\end{align*}
Also, we do not have $\int fd(\phi\oplus\psi )=\int fd\phi +\int fd\psi$. For example, suppose that
$\omega _1,\omega _2\in A$. Then
\begin{align*}
\int\chi _Ad(\omega _1^*\oplus\omega _2^*)&=(\omega _1^*\oplus\omega _2^*)(A)
  =0\ne\omega _1^*(A)+\omega _2^*(A)\\
  &=\int\chi _Ad\omega _1^*+\int\chi _Ad\omega _2^*
\end{align*}

It is frequently convenient to write \eqref{eq34} and \eqref{eq35} in a different form. For $f\colon\Omega\to\real ^+$, let $\omega _1,\ldots ,\omega _r$ be the points of $\Omega$ such that $f(\omega _i)>0$, $i=1,\ldots ,r$, with
$f(\omega _1)\le\cdots\le f(\omega _r)$. It easily follows from \eqref{eq34} that
\begin{align}         
\label{eq36}
\int fd\phi&=f(\omega _1)\phi\paren{\brac{\omega _1,\ldots ,\omega _r}}+\sqbrac{f(\omega _2)-f(\omega _1)}
  \phi\paren{\brac{\omega _2,\ldots ,\omega _r}}\notag\\
  &\quad +\cdots +\sqbrac{f(\omega _r)-f(\omega _{r-1})}\phi (\omega _r)\\
\label{eq37}         
&=\sum _{i=1}^rf(\omega _i)\sqbrac{\phi\paren{\brac{\omega _i,\ldots ,\omega _r}}
  -\phi\paren{\brac{\omega _{i+1},\ldots ,\omega _r}}}
\end{align}
In \eqref{eq36} we have used the shorthand notation $\phi (\omega _r)=\phi\paren{\brac{\omega _r}}$.

We now compute $q$-integrals relative to some of the common coevents. These formulas follow easily from \eqref{eq34} or \eqref{eq36} and their verification is left to the reader. First, it is not surprising that
$\int fd\omega ^*=f(\omega )$. If $0\le f(\omega _1)\le\cdots\le f(\omega _r)$ then
\begin{equation}         
\label{eq38}
\int fd(\omega _1^*\oplus\cdots\oplus\omega _r^*)=f(\omega _r)-f(\omega _{r-1})+\cdots +(-1)^{r+1}f(\omega _1)
\end{equation}
For example, if $f\colon\Omega\to\real ^+$ is arbitrary then 
\begin{equation*}
\int fd(\omega _1^*\oplus\omega _2^*)=\max\paren{f(\omega _1),f(\omega _2)}
  -\min\paren{f(\omega _1),f(\omega _2)}
\end{equation*}
For $f\colon\Omega\to\real ^+$ we have
\begin{align}         
\label{eq39}
\int fd(\omega _1^*\cdots\omega _r^*)&=\min\paren{f(\omega _1),\ldots ,f(\omega _r)}\\
\label{eq310}         
\int fdA^*&=\max\brac{f(\omega )\colon\omega\in A}
\end{align}
For $f\colon\Omega\to\real ^+$, letting $\alpha =\max\brac{f(\omega )\colon\omega\in\Omega}$,
$\beta =\min\brac{f(\omega )\colon\omega\in A}$ we have
\begin{equation}         
\label{eq311}
\int fdA_*=\begin{cases}\alpha -\beta&\text{if }f^{-1}(\alpha )\subseteq A\\
0&\text{if }f^{-1}(\alpha )\not\subseteq A\end{cases}
\end{equation}

\begin{exam}{4}    
We illustrate \eqref{eq38}--\eqref{eq311} by computing these $q$-integrals for specific examples. On $\Omega _5$ suppose $0<f(\omega _1)<\cdots <f(\omega _5)$. Applying \eqref{eq36} gives
\begin{align*}
\int fd(\omega _2^*\oplus\omega _3^*\oplus\omega _4^*)&=f(\omega _1)+\sqbrac{f(\omega _2)-f(\omega _1)}
  +\sqbrac{f(\omega _3)-f(\omega _2)}\cdot 0\\
  &\quad +\sqbrac{f(\omega _4)-f(\omega _3)}+\sqbrac{f(\omega _5)-f(\omega _4)}\cdot 0\\
  &=f(\omega _4)-f(\omega _3)+f(\omega _2)
\end{align*}
\begin{align*}
\int fd(\omega _2^*\omega _3^*\omega _4^*)&=f(\omega _1)+\sqbrac{f(\omega _2)-f(\omega _1)}
  +\sqbrac{f(\omega _3)-f(\omega _2)}\cdot 0\\
  &\quad +\sqbrac{f(\omega _4)-f(\omega _3)}\cdot 0+\sqbrac{f(\omega _5)-f(\omega _4)}\cdot 0\\
  &=f(\omega _2)=\min\paren{f(\omega _2),f(\omega _3),f(\omega _4)}
\end{align*}
which are \eqref{eq38} and \eqref{eq39}, respectively. Letting $A=\brac{\omega _2,\omega _3,\omega _4}$ we have
\begin{align*}
\int fdA^*&=f(\omega _1)+\sqbrac{f(\omega _2)-f(\omega _1)}+\sqbrac{f(\omega _3)-f(\omega _2)}
  +\sqbrac{f(\omega _4)-f(\omega _3)}\\
  &\quad +\sqbrac{f(\omega _5)-f(\omega _4)}\cdot 0=f(\omega _4)=\max\brac{f(\omega )\colon\omega\in A}
\end{align*}
\begin{align*}
\int fdA_*&=f(\omega _1)\cdot 0+\sqbrac{f(\omega _2)-f(\omega _1)}\cdot 0
  +\sqbrac{f(\omega _3)-f(\omega _2)}\cdot 0\\
  &\quad +\sqbrac{f(\omega _4)-f(\omega _3)}\cdot 0+\sqbrac{f(\omega _5)-f(\omega -4)}\cdot 0=0
\end{align*}
which are \eqref{eq310} and \eqref{eq311}, respectively. This last example is more interesting if we let
$B=\brac{\omega _3,\omega _4,\omega _5}$ in which case we have
\begin{align*}
\int fdB_*&=f(\omega _1)\cdot 0+\sqbrac{f(\omega _2)-f(\omega _1)}\cdot 0+\sqbrac{f(\omega _3)-f(\omega _2)}
  +\sqbrac{f(\omega _4)-f(\omega _3)}\\
  &\quad +\sqbrac{f(\omega _5)-f(\omega _4)}=f(\omega _5)-f(\omega _2)
\end{align*}
which are \eqref{eq311}.\hfill\qedsymbol
\end{exam}
\medskip

Let $f$ have the representation \eqref{eq32} and let $\phi _A$ be an atomic coevent for $A\in\ascript$. A straightforward application of \eqref{eq34} shows that
\begin{equation}         
\label{eq312}
\int fd\phi _A=\begin{cases}
\alpha _m-\alpha _{m-1}&\text{if }{\displaystyle A=\bigcupdotimn A_i}\text{ for some }m>1\\
\noalign{\smallskip}\alpha _1&\text{if }{\displaystyle A=\bigcupdotionen A_i}\\
\noalign{\smallskip}0&\text{otherwise}\end{cases}
\end{equation}
We mentioned previously that any $\phi\in\ascript ^*$ has a unique representation
$\phi =\phi _1\vee\cdots\vee\phi _m$ where $\phi _1,\ldots ,\phi _m$ are the distinct atoms below $\phi$. The next result shows that the $q$-integral with respect to $\phi$ is just the sum of the $q$-integrals with respect to the
$\phi _i$, $i=1,\ldots ,m$.

\begin{thm}       
\label{thm31}
If $\phi _1,\ldots ,\phi _m$ are distinct atoms in $\ascript ^*$ and $f\colon\Omega\to\real ^+$, then
\begin{equation*}
\int fd(\phi _1\vee\cdots\vee\phi _m)=\int fd\phi _1+\cdots +\int fd\phi _m
\end{equation*}
\end{thm}
\begin{proof}
We give the proof for $m=2$ and the general case is similar. Letting $f$ have the representation \eqref{eq32} and
$\phi _1=\phi _A$, $\phi _2=\phi _B$ for $A\ne B$, then \eqref{eq34} gives
\begin{align}         
\label{eq313}
\int fd(\phi _A\vee\phi _B)&=\alpha _1(\phi _A\vee\phi _B)\paren{\bigcupdotionen A_i}+(\alpha _2-\alpha _1)
  (\phi _A\vee\phi _B)\paren{\bigcupdotitwon A_i}\notag\\
  &\quad +\cdots +(\alpha _n-\alpha _{n-1})(\phi _A\vee\phi _B)(A_n)
\end{align}
If the right side of \eqref{eq313} is zero, then it is clear that
\begin{equation*}
\int fd\phi _A=\int fd\phi _B=0
\end{equation*}
and we are finished. Otherwise, at most two of the terms on the right side of \eqref{eq313} are nonzero. Suppose the
$j$th term is the only nonzero term. We can assume without loss of generality that
\begin{equation*}
\phi _A\paren{\bigcupdotijn A_i}=1\quad\text{and}\quad\phi _B\paren{\bigcupdotijn A_i}=0
\end{equation*}
Then $\int fd\phi _B=0$ and we have
\begin{equation*}
\int fd(\phi _A\vee\phi _B)=\alpha _j-\alpha _{j-1}=\int fd\phi _A+\int fd\phi _B
\end{equation*}
Next suppose the $j$th term and the $k$th term are nonzero. There are three cases one of which being
\begin{equation*}
\phi _A\paren{\bigcupdotijn A_i}=1,\ \phi _B\!\paren{\bigcupdotijn A_i}=0
  \ \ \text{and}\ \ \phi _B\!\paren{\bigcupdotikn A_i}=1,\ \phi _A\!\paren{\bigcupdotikn A_i}=0
\end{equation*}
In this case we have
\begin{equation*}
\int fd(\phi _A\vee\phi _B)=\alpha _j-\alpha _{j-1}+\alpha _k-\alpha _{k-1}=\int fd\phi _A+\int fd\phi _B
\end{equation*}
The other cases are similar.
\end{proof}

As usual in integration theory, for $A\in\ascript$ and $\phi\in\ascript ^*$ we define
\begin{equation*}
\int _Afd\phi =\int f\chi _Ad\phi
\end{equation*}
Simple integrals of this type are $\int _{\brac{\omega}}fd\phi =f(\omega )\phi (\omega )$ and $\int _A1d\phi =\phi (A)$. Other examples are
\begin{align*}
\int _Afd1&=\max\brac{f(\omega )\colon\omega\in A}\\
\int _AfdB^*&=\max\brac{f(\omega )\colon\omega\in A\cap B}
\end{align*}
Moreover, if $0\le f(\omega _1)\le f(\omega _2)$, then
\begin{equation*}
\int _{\brac{\omega _1,\omega _2}}fd\phi =f(\omega _1)\phi\paren{\brac{\omega _1,\omega _2}}
  +\sqbrac{f(\omega _2)-f(\omega _1)}\phi (\omega _2)
\end{equation*}

\begin{exam}{5}    
This example shows that the $q$-integral is not additive, even if $\phi$ is additive, in the sense that in general
\begin{equation*}
\int _{A\subcupdot B}fd\phi\ne\int _Afd\phi +\int _Bfd\phi
\end{equation*}
Let $A\cap B=\emptyset$, $\omega _1\in A$, $\omega _2\in B$ and $f\colon\Omega\to\real ^+$ with
$0<f(\omega _1)<f(\omega _2)$. We then have
\begin{align*}
\int _{A\subcupdot B}fd(\omega _1^*\oplus\omega _2^*)
  &=\int f\chi _{A\subcupdot B}d(\omega _1^*\oplus\omega _2^*)=f(\omega _2)-f(\omega _1)\\
  &\ne f(\omega _1)+f(\omega _2)
  =\int _Afd(\omega _1^*\oplus\omega _2^*)+\int _Bfd(\omega _1^*\oplus\omega _2^*)\hskip 1.75pc\qedsymbol
\end{align*}
\end{exam}    

\begin{exam}{6}    
This example shows that the $q$-integral is not grade-2 additive, even if $\phi$ is additive, in the sense that in general
\begin{equation*}
\int _{A\subcupdot B\subcupdot C}fd\phi\ne\int _{A\subcupdot B}fd\phi +\int _{A\subcupdot C}fd\phi
  +\int _{B\subcupdot C}fd\phi -\int _Afd\phi -\int _Bfd\phi -\int _Cfd\phi
\end{equation*}
Let $A,B,C\in\ascript$ be mutually disjoint, $\omega _1\in A$, $\omega _2\in B$, $\omega _3\in C$ and
$f\colon\Omega\to\real ^+$ with $0<f(\omega _1)<f(\omega _2)<f(\omega _3)$. For
$\phi =\omega _1^*\oplus\omega _2^*\oplus\omega _3^*$ we have
\begin{equation*}
\int _{A\subcupdot B\subcupdot C}fd\phi =f(\omega _3)-f(\omega _2)+f(\omega _1)
\end{equation*}
However,
\begin{align*}
\int _{A\subcupdot B}fd\phi&+\int _{A\subcupdot C}fd\phi +\int _{B\subcupdot C}fd\phi
  -\int _Afd\phi -\int _Bfd\phi -\int _Cfd\phi\\
  &=f(\omega _3)-f(\omega _2)-3f(\omega _1)\hskip 16.25pc\qedsymbol
\end{align*}
\end{exam}   

\section{Reality Filters} 
We interpret $\ascript ^*$ as the set of coevents that correspond to possible realities of our physical system. Presumably, there is only one ``actual reality'' $\phi _a$. But how do we find $\phi _a$? We need methods for filtering out the unwanted potential realities until we are left only with $\phi _a$. We have seen that
$\ab{\ascript ^*}=2^{(2^n-1)}$ which can be an inconceivably large number. Thus, finding $\phi _a$ is like finding a needle in a double exponential haystack. For example, for a small system with $n=\ab{\Omega}=6$ sample points, we have
\begin{equation*}
\ab{\ascript ^*}=2^{63}=9,223,372,036,854,775,808
\end{equation*}
If this were a classical system, we only have to sort through $6$ potential realities and this could be achieved by a simple observation or measurement. Filters that have been used in the past have been to assume that the potential realities must be additive or that they must be multiplicative. In these two cases there are $2^6=64$ potential realities which is quite manageable. Another filter that has been considered is the assumption that the potential realities must be quadratic in which case we have $2^{21}=1,097,152$ potential realities. Although large, this is a lot better than $2^{63}$.

We now discuss another method for eliminating unwanted coevents. Suppose there is an experimental or theoretical reason for assuming that $A\in\ascript$ does not occur (or that $A$ is false). We then say that $A$ is \textit{precluded} \cite{sor94, sor07, sor09}. A trivial example is the empty set $\emptyset$ which by definition is precluded. Denoting the set of precluded events by $\ascript _0$ we call $\ascript _p=\ascript\smallsetminus\ascript _0$ the set of permitted events. A coevent $\phi$ is \textit{preclusive} if $\phi (A)=0$ for all $A\in\ascript _0$ \cite{sor94, sor07, sor09} and the set of preclusive coevents
\begin{equation*}
\ascript _p^*=\brac{\phi\in\ascript ^*\colon\phi (A)=0\text{ for all }A\in\ascript _0}
\end{equation*}
is the \textit{preclusive anhomomorphic logic}. Although $\ascript _p$ is not a Boolean algebra, $\ascript _p^*$ is a Boolean algebra and its atoms are the preclusive atoms of $\ascript ^*$. Hence,
$\ab{\ascript _p^*}=2^{\ab{\ascript _p}}$. For example, in $\Omega _2$ if $\ascript _0=\brac{\brac{\omega _1}}$, then
$\ascript _p=\brac{\brac{\omega _2},\Omega _2}$ and $\ab{\ascript _p^*}=2^2=4$. We then have
\begin{equation*}
\ascript _p^*=\brac{0,\omega _2^*,\omega _1^*\omega _2^*,\omega _2^*\oplus\omega _1^*\omega _2^*}
\end{equation*}
If $\ascript _0=\brac{\Omega _2}$, then $\ascript _p=\brac{\brac{\omega _1},\brac{\omega _2}}$ and
$\ab{\ascript _p^*}=4$. We then have
\begin{equation*}
\ascript _p^*=\brac{0,\omega _1^*\oplus\omega _1^*\omega _2^*,\omega _2^*\oplus\omega _1^*\omega _2^*,
  \omega _1^*\oplus\omega _2^*}
\end{equation*}
If $\ascript _0=\brac{\brac{\omega _1},\brac{\omega _2}}$, then $\ascript _p=\brac{\Omega _2}$,
$\ab{\ascript _p^*}=2$ and $\ascript _p^*=\brac{0,\omega _1^*\omega _2^*}$.

A coevent $\phi$ is \textit{regular} if
\begin{list} {(R\arabic{cond})}{\usecounter{cond}
\setlength{\rightmargin}{\leftmargin}}
\item $\phi (A)=0$ implies $\phi (A\cupdot B)=\phi (B)$ for all $B\in\ascript$
\item $\phi (A\cupdot B)=0$ implies $\phi (A)=\phi (B)$
\end{list}
Condition (R1) is a reasonable condition and although (R2) is not as clear, it may have merit. (We will mention later that quantum measures are usually assumed to be regular.) In any case, there may be fundamental reasons for assuming that the actual reality is regular which gives another method for filtering out unwanted potential realities. For example the regular coevents on $\Omega _2$ are $0,1,\omega _1^*,\omega _2^*$ and
$\omega _1^*\oplus\omega _2^*$.

The filter $F$ that we now discuss is much stronger than the previous ones. In fact, it is so strong that it may filter out all of $\ascript _p^*$ in which case we say it is not successful. In all the examples we have considered, when $F$ is successful then $F$ admits a unique reality. We believe that if $F$ does not produce a unique coevent, then the number of coevents it does produce is very small. To describe $F$ we shall need the concept of a quantum measure ($q$-measure, for short) \cite{gud1, gud2, sor94, sor07}. A $q$-\textit{measure} is a set function
$\mu\colon\ascript\to\real ^+$ that satisfies the \textit{grade}-2 \textit{additivity condition}
\begin{equation}         
\label{eq41}
\mu (A\cupdot B\cupdot C)=\mu (A\cupdot B)+\mu (A\cupdot C)+\mu (B\cupdot C)-\mu (A)-\mu (B)-\mu (C)
\end{equation}
Because of quantum interference, $\mu$ may not satisfy the \textit{grade}-1 \textit{additivity condition}
$\mu (A\cupdot B)=\mu (A)+\mu (B)$ that holds for ordinary measures. Of course, grade-1 additivity implies grade-2 additivity but the converse does not hold \cite{gt09, gud1, sal02}. A $q$-measure $\mu$ is \textit{regular} if $\mu$ satisfies (R1) and (R2) (with $\phi$ replaced by $\mu$) and it is usually assumed that $q$-measures are regular. For generality, we do not make that assumption here. Since an ordinary measure is grade-1 additive, it is determined by its values on singleton sets. In a similar way, a $q$-measure is determined by its values on singleton and doubleton sets. In fact, by \eqref{eq41} we have
\begin{align*}
\mu&\paren{\brac{\omega _1,\omega _2,\omega _3}}\\
  &=\mu\paren{\brac{\omega _1,\omega _2}}+\mu\paren{\brac{\omega _1,\omega _3}}
  +\mu\paren{\brac{\omega _2,\omega _3}}-\mu (\omega _1)-\mu (\omega _2)-\mu (\omega _3)
\end{align*}
and continuing by induction we obtain
\begin{equation}         
\label{eq42}
\mu\paren{\brac{\omega _1,\ldots ,\omega _m}}=\sum _{i<j=1}^m\mu\paren{\brac{\omega _i,\omega _j}}
  -(m-2)\sum _{i=1}^m\mu (\omega _i)
\end{equation}

There are reasons to believe that a (finite) quantum system can be described by a $q$-measure space
$(\Omega ,\ascript ,\mu )$ where $\ab{\Omega}<\infty$, $\ascript =2^\Omega$ and $\mu\colon\ascript\to\real ^+$ is a fixed $q$-measure that is specified by nature \cite{gt09, gud1, sor94, sor07}. A $q$-measure $\mu$ on $\ascript$
1-\textit{generates} a coevent $\phi\in\ascript ^*$ if there exists a strictly positive function $f\colon\Omega\to\real$ such that $\mu (A)=\int _Afd\mu$ for all $A\in\ascript$. We call $f$ a $\phi$-\textit{density} for$\mu$. In a sense, $\mu$ is an ``average'' of the density $f$ with respect to the potential reality $\phi$. Put another way, $\mu$ is an ``average'' of the truth values of $\phi$. There are some immediate questions that one might ask.
\begin{list} {(Q\arabic{cond})}{\usecounter{cond}
\setlength{\rightmargin}{\leftmargin}}
\item Does every $q$-measure 1-generate at least one coevent?
\item Is every coevent 1-generated by at least one $q$-measure?
\item If $\mu$ 1-generates $\phi$, is $\phi$ unique?
\item If $\phi$ is 1-generated by $\mu$, is $\mu$ unique?
\item If $f$ is a $\phi$-density for $\mu$, is $f$ unique?
\end{list}
We shall show that the answers to (Q1), (Q2), (Q4) and (Q5) are no. We do not know the answer to (Q3) although we have a partial result.

The definition of $\mu$ 1-generating a coevent $\phi$ is quite simple and if $\phi$ is unique that's great, but we shall give examples of $q$-measures that do not 1-generate any coevent. One reason for this is that a function
$f\colon\Omega _n\to\real$ gives at most $n$ pieces of information while a $q$-measure is determined by its values on singleton and doubleton sets so $n(n+1)/2$ pieces of information may be needed. For this reason, we introduce a more complicated (and presumably more general) definition. A function $f\colon\Omega\times\Omega\to\real$ is \textit{symmetric} if $f(\omega ,\omega ')=f(\omega ',\omega )$ for all $\omega ,\omega '\in\Omega$. Notice that a symmetric function on $\Omega _n\times\Omega _n$ has $n(n+1)/2$ possible values. A $q$-measure $\mu$ on
$\ascript$ 2-\textit{generates} a coevent $\phi\in\ascript ^*$ if there exists a strictly positive symmetric function
$f\colon\Omega\times\Omega\to\real$ such that
\begin{equation*}
\mu (A)=\int _A\sqbrac{\int _Af(\omega ,\omega ')d\phi (\omega )}d\phi (\omega ')
\end{equation*}
for every $A\in\ascript$. We again call $f$ a $\phi$-\textit{density} for $\mu$. It is interesting to note that a
$\phi$-density $f$ determines a symmetric matrix $\sqbrac{f(\omega _i,\omega _j)}$ with positive entries which reminds us of a density matrix or state but we have not found this observation useful. The 2-generation of coevents is what we referred to previously as the strong reality filter $F$. In Section~5 we shall present some general results involving 1 and 2-generation of coevents and in Sections~6 and 7 we discuss specific examples. We also consider questions (Q1)--(Q5) for 2-generation of coevents.

The reason we require the density $f$ to be strictly positive is that if $f$ is allowed to be zero, then $\phi$ may be generated by a $q$-measure in a trivial way. For example, suppose that $\phi (\omega _0)=1$ and let
$f=\chi _{\brac{\omega _0}}$. Then for every $A\in\ascript$ we have $\int _Afd\phi =\delta _{\omega _0}(A)$. Thus, the Dirac measure $\delta _{\omega _0}$ 1-generates $\phi$ in a trivial manner. In this sense, $\delta _{\omega _0}$
``generates'' any $\phi\in\ascript ^*$ satisfying $\phi (\omega _0)=1$ so $\phi$ is highly nonunique. However, suppose that $f$ is strictly positive and $\delta _{\omega _0}(A)=\int _Afd\phi$ for every $A\in\ascript$. We then have
\begin{equation*}
1=\int _{\brac{\omega _0}}fd\phi =f(\omega _0)\phi (\omega _0)
\end{equation*}
so that $f(\omega _0)=\phi (\omega _0)=1$. Let $A=\brac{\omega _1,\ldots ,\omega _r}$ with
$0<f(\omega _1)\le\cdots\le f(\omega _r)$. If $\omega _0\notin A$ then by \eqref{eq36} we have
$f(\omega _1)\phi (A)=0$ so $\phi (A)=0$. If $\omega _0\in A$, then $\omega _0=\omega _i$ for some
$i\in\brac{1,\ldots ,r}$. Hence, by \eqref{eq36} we have
\begin{align*}
1=\int _Afd\phi&=f(\omega _1)\phi (A)
  +\sqbrac{f(\omega _2)-f(\omega _1)}\phi\paren{\brac{\omega _2,\ldots ,\omega _r}}\\
  &\quad +\cdots +\sqbrac{f(\omega _i)-f(\omega _{i-1})}\phi\paren{\brac{\omega _i,\ldots ,\omega _r}}
\end{align*}
If $\phi (A)=0$, then
\begin{align*}
1&\le f(\omega _2)-f(\omega _1)+f(\omega _3)-f(\omega _2)+\cdots +f(\omega _0)-f(\omega _{i-1})\\
  &=f(\omega _0)-f(\omega _1)=1-f(\omega _1)<1
\end{align*}
which is a contradiction. Hence, $\phi (A)=1$. It follows that $\phi =\omega _0^*$ is the unique coevent 1-generated by $\delta _{\omega _0}$.

We now show that if $\phi$ is 1-generated by $\mu$ and $\phi (A)\ne 0$ whenever $\mu (A)\ne 0$, then $\phi$ is
2-generated by $\mu$. Indeed, suppose that $\mu (A)=\int _Afd\phi$ for a strictly positive function
$f\colon\Omega\to\real$. Then $g(\omega ,\omega ')=\tfrac{1}{2}\sqbrac{f(\omega )+f(\omega ')}$ is a strictly positive symmetric function on $\Omega\times\Omega$ and we have
\begin{align*}
\int _A\sqbrac{\int _Ag(\omega ,\omega ')d\phi (\omega )}d\phi (\omega ')&=\frac{1}{2}\int _A
\sqbrac{\int _A\paren{f(\omega )+f(\omega ')}d\phi (\omega )}d\phi (\omega ')\\\noalign{\smallskip}
&=\frac{1}{2}\int _A\sqbrac{\mu (A)+f(\omega ')\phi (A)}d\phi (\omega ')\\
&=\mu (A)\phi (A)=\mu (A)
\end{align*}
Hence, $\phi$ is 2-generated by $\mu$. We do not know if $\phi$ 1-generated implies $\phi$ 2-generated, in general.

\section{General Results} 
If $\mu$ is a $q$-measure on $\ascript$ and $\phi\in\ascript ^*$ we say that $\phi$ is $\mu$-\textit{preclusive} if
$\phi (A)=0$ whenever $\mu (A)=0$. The following result shows that generation is stronger than preclusivity.

\begin{thm}       
\label{thm51}
If $\phi$ is 1-generated or 2-generated by $\mu$, then $\phi$ is $\mu$-preclusive.
\end{thm}
\begin{proof}
Suppose that $\phi$ is 1-generated by $\mu$ and $\mu (A)=\int _Afd\phi$ for all $A\in\ascript$. If
$A=\brac{\omega _1,\ldots ,\omega _r}$ with $\mu (A)=0$, then by \eqref{eq36} we have
\begin{equation*}
f(\omega _1)\phi (A)+\sqbrac{f(\omega _2)-\!f(\omega _1)}\phi\paren{\brac{\omega _2,\ldots ,\omega _r}}
   +\cdots +\sqbrac{f(\omega _r)-f(\omega _{r-1})}\mu (\omega _r)\!=\!0
\end{equation*}
Since $f(\omega _1)>0$ we conclude that $\phi (A)=0$. Next suppose that $\phi$ is 2-generated by $\mu$ and
\begin{equation*}
\mu (A)=\int _A\sqbrac{\int _Af(\omega ,\omega ')d\phi (\omega )}d\phi (\omega ')
\end{equation*}
for all $A\in\ascript$. Assume that $\mu (A)=0$ but $\phi (A)=1$. Let
\begin{equation*}
g(\omega ')=\int _Af(\omega ,\omega ')d\phi (\omega )
\end{equation*}
If $m(\omega ')=\min\brac{f(\omega ,\omega ')\colon\omega\in A}$, then $m(\omega ')>0$ and by \eqref{eq36}
$g(\omega ')\ge m(\omega ')$ for all $\omega '\in A$. Letting $m=\min\brac{m(\omega ')\colon\omega '\in A}$ we have that $g(\omega ')\ge m>0$ for all $\omega '\in A$. But then
\begin{equation*}
\mu (A)=\int _Ag(\omega ')d\phi (\omega ')>0
\end{equation*}
This is a contradiction so $\phi (A)=0$.
\end{proof}

The next theorem is a partial uniqueness result which shows that within the set of regular coevents 1-generated coevents are unique.

\begin{thm}       
\label{thm52}
If $\phi ,\psi\in\ascript ^*$ are regular and $\mu$ 1-generates both $\phi$ and $\psi$, then $\phi =\psi$.
\end{thm}
\begin{proof}
We have that $\int _Afd\phi =\int _Agd\psi$ for all $A\in\ascript$ where $f$ and $g$ are strictly positive. If
$\phi (\omega )=0$, then
\begin{equation*}
\int _{\brac{\omega}}fd\phi =0=\int _{\brac{\omega}}gd\psi =g(\omega )\psi (\omega )
\end{equation*}
so $\psi (\omega )=0$. By symmetry $\phi (\omega )=0$ if and only if $\psi (\omega )=0$ so $\phi$ and $\psi$ agree on all singleton sets. If $\phi (\omega )=1$ so $\psi (\omega )=1$, then 
\begin{equation}         
\label{eq51}
f(\omega )=\int _{\brac{\omega}}fd\phi =\int _{\brac{\omega}}gd\phi =g(\omega )
\end{equation}
Suppose $\phi\paren{\brac{\omega _1\omega _2}}=0$. By regularity $\phi (\omega _1)=\phi (\omega _2)=1$ or
$\phi (\omega _1)=\phi (\omega _2)=0$. In the second case $\psi (\omega _1)=\psi (\omega _2)=0$ so by regularity
\begin{equation*}
\psi\paren{\brac{\omega _1,\omega _2}}=\phi\paren{\brac{\omega _1,\omega _2}}=0
\end{equation*}
Suppose $\phi (\omega _1)=\phi (\omega _2)=1$. Then $\psi (\omega _1)=\psi (\omega _2)=1$ and by \eqref{eq51} we have $f(\omega _1)=g(\omega _1)$ and $f(\omega _2)=g(\omega _2)$. Suppose that
$f(\omega _1)\le f(\omega _2)$. Then if $\psi\paren{\brac{\omega _1,\omega _2}}=1$ we have
\begin{equation*}
\int _{\brac{\omega _1,\omega _2}}\!\!\!\!\!\!\!\!\!\!fd\phi =f(\omega _2)-f(\omega _1)
=\int _{\brac{\omega _1,\omega _2}}\!\!\!\!\!\!\!\!\!\!gd\psi
  =g(\omega _1)+g(\omega _2)-g(\omega _1)=g(\omega _2)=f(\omega _2)
\end{equation*}
which implies $f(\omega _1)=0$, a contradiction. Hence, $\phi$ and $\psi$ agree on all doubleton sets. Suppose
$\phi\paren{\brac{\omega _1,\omega _2,\omega _3}}=0$ and
$\psi\paren{\brac{\omega _1,\omega _2,\omega _3}}=1$.
If $\phi\paren{\brac{\omega _2,\omega _3}}=\phi (\omega _1)=0$, then
$\psi\paren{\brac{\omega _2,\omega _3}}=\psi (\omega _1)=1$. If $\phi (\omega _3)=0$ we obtain by regularity that
\begin{equation*}0=\phi\paren{\brac{\omega _1,\omega _2,\omega _3}}=\phi\paren{\brac{\omega _2,\omega _3}}=1
\end{equation*}
which is a contradiction. Hence, $\phi (\omega _3)=1$. Suppose that $f(\omega _1)\le f(\omega _2)\le f(\omega _3)$. Then
\begin{align*}
f(\omega _3)-f(\omega _1)&=f(\omega _2)-f(\omega _1)+f(\omega _3)-f(\omega _2)\\
  &=\int _{\brac{\omega _1,\omega _2,\omega _3}}\!\!\!\!\!\!\!\!\!\!fd\phi
  =\int _{\brac{\omega _1,\omega _2,\omega _3}}\!\!\!\!\!\!\!\!\!\!gd\psi\\
  &=g(\omega _1)+g(\omega _2)-g(\omega _1)+g(\omega _3)-g(\omega _2)\\
  &=g(\omega _3)=f(\omega _3)
\end{align*}
Hence, $f(\omega _1)=0$ which is a contradiction. We conclude that $\phi$ and $\psi$ agree on all tripleton sets. Continue this process by induction.
\end{proof}

For a set function $\nu\colon\ascript _m\to\real ^+$ we define the \textit{expansion}
$\nuhat\colon\ascript _n\to\real ^+$, $n\ge m$, of $\nu$ by $\nuhat (A)=\nu (A\cap\Omega _m)$. As a special case, if
$\phi\colon\ascript _m\to\integers _2$, the expansion $\phihat\colon\ascript _n\to\integers _2$ of $\phi$ is given by
$\phihat (A)=\phi (A\cap\Omega _m)$. It is clear that if $\phi\in\ascript _m^*$, then $\phihat\in\ascript _n^*$. Also, if
$\phi$ has the evaluation map representation
\begin{equation*}
\phi =\omega _1^*\oplus\cdots\oplus\omega _r^*\oplus\omega _s^*\omega _t^*
  \oplus\cdots\oplus\omega _u^*\omega _v^*\omega _w^*\oplus\cdots
\end{equation*}
on $\ascript _m$ then $\phihat$ has the same representation on $\ascript _n$. It is also clear that $\nu =\nuhat\mid\ascript _m$ where $\nuhat\mid\ascript _m$ is the restriction of $\nuhat$ to $\ascript _m$. The proof of the next lemma is straightforward.

\begin{lem}       
\label{lem53}
{\rm (a)}\enspace $\mu$ is a $q$-measure on $\ascript _m$ if and only if $\muhat$ is a $q$-measure on $\ascript _n$.
{\rm (b)}\enspace $\mu$ is a regular $q$-measure on $\ascript _m$ if and only if $\muhat$ is a regular $q$-measure
on $\ascript _n$.
\end{lem}

We call the next result the expansion theorem.

\begin{thm}       
\label{thm54}
{\rm (a)}\enspace $\mu$ 1-generates $\phi$ if and only if $\muhat$ 1-generates $\phihat$
{\rm (b)}\enspace $\mu$ 2-generates $\phi$ if and only if $\muhat$ 2-generates $\phihat$.
\end{thm}
\begin{proof}
(a)\enspace Suppose $\mu (A)=\int _Ad\phi$ for all $A\in\ascript$, where $f\colon\Omega _m\to\real$ is strictly positive. For $n>m$, define $\fhat\colon\Omega _n\to\real$ by
\begin{equation*}
\fhat (\omega _i)=\begin{cases}f(\omega _i)&\textit{if }\omega _i\in\Omega _m\\
M&\textit{if }\omega _i\in\Omega _n\smallsetminus\Omega _m\end{cases}
\end{equation*}
where $M=\max\brac{f(\omega _i)\colon\omega _i\in\Omega _m}$. Then $\fhat\colon\Omega _n\to\real$ is strictly
positive. Let $A=\brac{\omega _1,\ldots ,\omega _r,\omega _{r+1},\ldots ,\omega _s}$ where $\omega _1,\ldots ,\omega _r\in\Omega _m$ and $\omega _{r+1},\ldots ,\omega _s\in\Omega _n\smallsetminus\Omega _m$. We can assume without loss of generality that
\begin{equation*}
\fhat (\omega _1)\le\fhat (\omega _2)\le\cdots\le\fhat (\omega _r)\le M
  =\fhat (\omega _{r+1})=\cdots =\fhat (\omega _s)
\end{equation*}
Then by \eqref{eq36} we have
\begin{align*}
\int _A\fhat d\phihat&=\fhat (\omega _1)\phihat (A)+\sqbrac{\fhat (\omega _2)-\fhat (\omega _1)}
  \phihat\paren{\brac{\omega _2,\ldots ,\omega _s}}\\
  &\quad +\cdots +\sqbrac{\fhat (\omega _r)-\fhat (\omega _{r-1})}\phihat\paren{\brac{\omega _r,\ldots ,\omega _s}}\\
  &=f(\omega _1)\phi (A\cap\Omega _m)+\sqbrac{f(\omega _2)-f(\omega _1)}
  \phi\paren{\brac{\omega _2,\ldots ,\omega _r}}\\
  &\quad +\cdots +\sqbrac{f(\omega _r)-f(\omega _{r-1})}\phi (\omega _r)\\
  &=\int _{A\cap\Omega _m}fd\phi =\mu (A\cap\Omega _m)=\muhat (A)
\end{align*}
Hence, $\muhat$ 1-generates $\phihat$. The converse is straightforward.\newline
(b)\enspace Suppose
\begin{equation*}
\mu (A)=\int _A\sqbrac{\int _Af(\omega ,\omega ')d\phi (\omega )}d\phi (\omega ')
\end{equation*}
for all $A\in\ascript _m$ where $f\colon\Omega _m\times\Omega _m\to\real$ is strictly positive and symmetric. For $n>m$ define $\fhat\colon\Omega _n\times\Omega _n\to\real$ by
\begin{equation*}
\fhat (\omega _i,\omega _j)=
\begin{cases}f(\omega _i,\omega _j)&\text{if }(\omega _i,\omega _j)\in\Omega _m\times\Omega _m\\
M&\text{if }(\omega _i,\omega _j)\in\Omega _n\times\Omega _n\smallsetminus\Omega _m\times\Omega _m
\end{cases}
\end{equation*}
where $M=\max (M_1,M_2)$ and
\begin{align*}
M_1&=\max\brac{f(\omega _i,\omega _j)\colon (\omega _i,\omega _j)\in\Omega _m\times\Omega _m}\\
M_2&=\max\brac{\int _Af(\omega ,\omega _i)d\phi (\omega )\colon\omega _i\in\Omega _m,A\in\ascript _m}
\end{align*}
Then $\fhat$ is strictly positive and symmetric on $\Omega _n\times\Omega _n$. Again, let
$A=\brac{\omega _1,\ldots ,\omega _r,\omega _{r+1},\ldots ,\omega _s}$ where
$\omega _1,\ldots ,\omega _r\in\Omega _m$ and
$\omega _{r+1},\ldots ,\omega _s\in\Omega _n\smallsetminus\Omega _m$. Assume for simplicity that
\begin{equation*}
\fhat (\omega _1,\omega _1)\le\cdots\le\fhat (\omega _1,\omega _r)\le\fhat (\omega _2,\omega _3)\le\cdots\le
\fhat (\omega _2,\omega _r)\le\cdots\le\fhat (\omega _r,\omega _r)
\end{equation*}
and the other cases will be similar. Letting
\begin{equation*}
g(\omega ')=\int _A\fhat (\omega ,\omega ')d\phihat (\omega )
\end{equation*}
we have as before that
\begin{equation*}
g(\omega _i)=\int _{A\cap\Omega _m}f(\omega ,\omega _i)d\phi (\omega )
\end{equation*}
$i=1,\ldots ,r$, and
\begin{equation*}
g(\omega _{r+1})=\cdots =g(\omega _s)=M\phi (A\cap\Omega _m)
\end{equation*}
As in (a) we obtain
\begin{align*}
\int _A&\sqbrac{\int _A\fhat (\omega ,\omega ')d\phihat (\omega )}d\phihat (\omega ')
  =\int _Ag(\omega ')d\phihat (\omega ')=\int _{A\cap\Omega _m}g(\omega ')d\phi (\omega ')\\
  &=\int _{A\cap\Omega _m}\sqbrac{\int _{A\cap\Omega _m}f(\omega ,\omega ')d\phi (\omega )}d\phi (\omega ')
  =\mu (A\cap\Omega _m)=\muhat (A)
\end{align*}
Hence, $\muhat$ 2-generates $\phihat$. Again, the converse is straightforward.
\end{proof}

Notice that in the proof of Theorem~\ref{thm54} we could replace $M$ by any larger number and the result would be the same. This is one of many examples which show that the density function need not be unique. The expansion theorem can be quite useful. For example, in Section~7 we shall show that
$\phi =\omega _1^*\oplus\omega _2^*\oplus\omega _3^*\oplus\omega _1^*\omega _2^*$ is 2-generated by a
$q$-measure on $\ascript _3$. It follows from the expansion theorem that on $\ascript _n$ for $n>3$, $\phi$ is 2-generated by $\muhat$. Moreover, the restriction of $\phi$ to $\ascript _2$, namely
$\omega _1^*\oplus\omega _2^*\oplus\omega _1^*\omega _2^*$ is 2-generated.

We now present a useful lemma that determines values of a density function when $\phi (\omega )\ne 0$.

\begin{lem}       
\label{lem55}
{\rm (a)}\enspace Suppose $\mu$ 1-generates $\phi$ with $\phi$-density $f$. Then $\phi (\omega )=0$ if and only if
$\mu (\omega )=0$ and $\phi (\omega )\ne 0$ if and only if $f(\omega )=\mu (\omega )$.
{\rm (b)}\enspace Suppose $\mu$ 2-generates $\phi$ with $\phi$-density $f$. Then $\phi (\omega )=0$ if and only if
$\mu (\omega )=0$ and $\phi (\omega )\ne 0$ if and only if $f(\omega ,\omega )=\mu (\omega )$.
\end{lem}
\begin{proof}
(a)\enspace Since $\mu$ 1-generates $\phi$, we have
\begin{equation*}
\mu (\omega )=\int _{\brac{\omega}}fd\phi =f(\omega )\phi (\omega )
\end{equation*}
Hence, $\mu (\omega )=0$ if and only if $\phi (\omega )=0$. Moreover, if $\phi (\omega )\ne 0$, then 
$f(\omega )=\mu (\omega )$. Conversely, if $f(\omega )=\mu (\omega )$, then
\begin{equation*}
f(\omega )=\int _{\brac{\omega}}fd\phi =f(\omega )\phi (\omega )
\end{equation*}
Since $f(\omega )\ne 0$, $\phi (\omega )=1$.\newline
(b)\enspace Since $\mu$ 2-generates $\phi$, we have
\begin{align*}\mu (\omega _0)
&=\int _{\brac{\omega _0}}\sqbrac{\int _{\brac{\omega _0}}f(\omega ,\omega ')d\phi (\omega )}d\phi (\omega ')
=\int _{\brac{\omega _0}}f(\omega _0,\omega ')\phi (\omega _0)d\phi (\omega ')\\
&=f(\omega _0,\omega _0)\phi (\omega _0)
\end{align*}
Hence $\mu (\omega _0)=0$ if and only if $\phi (\omega _0)=0$. Moreover, if $\phi (\omega _0)\ne 0$, then
$f(\omega _0,\omega _0)=\mu (\omega _0)$. Conversely, if $f(\omega _0,\omega _0)=\mu (\omega _0)$ then
\begin{equation*}
f(\omega _0,\omega _0)=\int _{\brac{\omega _0}}
  \sqbrac{\int _{\brac{\omega _0}}f(\omega _0,\omega ')d\phi (\omega )}d\phi (\omega ')
  =f(\omega _0,\omega _0)\phi (\omega _0)
\end{equation*}
Since $f(\omega _0,\omega _0)\ne 0$, $\phi (\omega _0)=1$.
\end{proof}

The next two results show that 1 and 2-generation are strictly quantum phenomena except for a few simple cases.

\begin{thm}       
\label{thm56}
If $\mu$ is a (grade-1) measure on $\ascript$ that is not a Dirac measure $c\delta _\omega$, then $\mu$ is not 1-generating.
\end{thm}
\begin{proof}
Since $\mu$ is not a Dirac measure, there exist $\omega _1,\omega _2\in\Omega$ such that
$\mu (\omega _1),\mu (\omega _2)>0$. Suppose $\mu$ 1-generates $\phi\in\ascript ^*$ so that
$\mu (A)=\int _Afd\phi$ for all $A\in\ascript$ where $f\colon\Omega\to\real$ is strictly positive. It follows from
Lemma~\ref{lem55} that $f(\omega _i)=\mu (\omega _i)$ and $\phi (\omega _i)=1$, $i=1,2$. We can assume without loss of generality that $f(\omega _1)\le f(\omega _2)$. Since $\mu$ is a measure, we have
\begin{align*}
\mu (\omega _1)+\mu (\omega _2)&=\mu\paren{\brac{\omega _1,\omega _2}}\\
   &=\int _{\brac{\omega _1,\omega _2}}\!\!\!\!\!\!fd\phi =f(\omega _1)\phi\paren{\brac{\omega _1,\omega _2}}
   +\sqbrac{f(\omega _2)-f(\omega _1)}\phi (\omega _2)\\
   &=\mu (\omega _1)\phi\paren{\brac{\omega _1,\omega _2}}+\mu (\omega _2)-\mu (\omega _1)
\end{align*}
Hence, $\phi\paren{\brac{\omega _1,\omega _2}}=2$ which is a contradiction.
\end{proof}

In the sequel, we shall use the notation
\begin{equation*}
g_A(\omega ')=\int _Af(\omega ,\omega ')d\phi (\omega )
\end{equation*}
Of course, $g_A$ depends on $f$ and $\phi$ but these will be known by context.

\begin{thm}       
\label{thm57}
{\rm (a)}\enspace Any measure of the form $\mu =a_1\delta _{\omega _1}+a_2\delta _{\omega _2}$, $a_1,a_2>0$, 2-generates the coevent $\phi =\omega _1^*\oplus\omega _2^*\oplus\omega _1^*\omega _2^*$. In particular, any measure on $\ascript _2$ 2-generates $\omega ^*\in\ascript _2^*$ or $1\in\ascript _2^*$.
{\rm (b)}\enspace If $\mu$ is a measure on $\ascript$ that is not of the form
$a_1\delta _{\omega _1}+a_2\delta _{\omega _2}$, then $\mu$ is not 2-generating.
\end{thm}
\begin{proof}
(a)\enspace Define the strictly positive, symmetric function $f$ on $\Omega\times\Omega$ by
$f(\omega _i,\omega _i)=a_i$, $i=1,2$,
\begin{equation*}
f(\omega _1,\omega _2)=f(\omega _2,\omega _1)=a_1+a_2
\end{equation*}
and $f(\omega _i,\omega _j)=M$ otherwise, where $M>a_1+a_2$. For $A\in\ascript$, if
$\omega _1,\omega _2\notin A$, then $g_A(\omega ')=M\phi (A)=0$. Hence,
\begin{equation*}
\mu '(A)=\int _Ag(\omega ')d\phi (\omega ')=0
\end{equation*}
If $\omega _1\in A$, $\omega _2\notin A$, then
\begin{equation*}
g_A(\omega _1)=\int _Af(\omega ,\omega _1)d\phi (\omega )=a_1
\end{equation*}
and for $\omega '\ne\omega _1$ we have
\begin{equation*}
g_A(\omega ')=\int _Af(\omega ,\omega ')d\phi (\omega )=0
\end{equation*}
Hence,
\begin{equation*}
\mu '(A)=\int _Ag_A(\omega ')d\phi (\omega ')=a_1
\end{equation*}
Similarly, if $\omega _2\in A$, $\omega _1\notin A$, then $\mu '(A)=a_2$. If $\omega _1,\omega _2\in A$, then
\begin{align*}
g_A(\omega _1)&=\int _Af(\omega ,\omega _1)d\phi (\omega )=a_1+a_2\\
g_A(\omega _2)&=\int _Af(\omega ,\omega _2)d\phi (\omega )=a_2+a_1
\end{align*}
If $\omega '\ne\omega _1,\omega _2$, then
\begin{equation*}
g_A(\omega ')=\int _Af(\omega ,\omega ')d\phi (\omega ')=M
\end{equation*}
Hence,
\begin{equation*}
\mu '(A)=\int _Ag_A(\omega ')d\phi (\omega ')=a_1+a_2
\end{equation*}
We conclude that
\begin{equation*}
\mu (A)=\mu (A')=\int _A\sqbrac{\int _Af(\omega ,\omega ')d\phi (\omega )}d\phi (\omega ')
\end{equation*}
so $\mu$ 2-generates $\phi$.\newline
(b)\enspace Since $\mu$ is not of the form $a_1\delta _{\omega _1}+a_2\delta _{\omega _2}$, there exists
$\omega _1,\omega _2,\omega _3\in\Omega$ such that $\mu (\omega _1),\mu (\omega _2),\mu (\omega _3)>0$. We can assume without loss of generality that $\mu (\omega _1)\le\mu (\omega _2)\le\mu (\omega _3)$. Suppose $\mu$ 2-generates $\phi\in\ascript ^*$ with density $f$. It follows from Lemma~\ref{lem55} that
$f(\omega _i,\omega _i)=\mu (\omega _i)$ and that $\phi (\omega _i)=1$, $i=1,2,3$. We now have three cases.
\newline
\textbf{Case 1.}\enspace $\mu (\omega _1)\le\omega _2)\le f(\omega _1,\omega _2)$. Letting
$A=\brac{\omega _1,\omega _2}$, $\phi\paren{\brac{\omega _1,\omega _2}}=a$ we have that
\begin{align*}
g_A(\omega _1)&=\int _Af(\omega ,\omega _1)d\phi (\omega )=\mu (\omega _1)a
  +f(\omega _1,\omega _2)-\mu (\omega _1)\\
g_A(\omega _2)&=\int _Af(\omega ,\omega _2)d\phi (\omega )=\mu (\omega _2)a
  +f(\omega _1,\omega _2)-\mu (\omega _2)
\end{align*}
Since $\mu$ is a measure, we have
\begin{align*}
\mu (\omega _1)+&\mu (\omega _2)=\mu (A)=\int _Ag_A(\omega ')d\phi (\omega ')\\
  &=\sqbrac{\mu (\omega _2)(a-1)+f(\omega _1,\omega _2)}a+\sqbrac{\mu (\omega _1)-\mu (\omega _2)}a
  +\mu (\omega _2)-\mu (\omega _1)
\end{align*}
If $a=0$, then
\begin{equation*}
\mu (\omega _1)+\mu (\omega _2)=\mu (\omega _2)-\mu (\omega _1)
\end{equation*}
which is a contradiction. If $a=1$, then $f(\omega _1,\omega _2)=\mu (\omega _1)+\mu (\omega _2)$.

\medskip
\noindent\textbf{Case 2.}\enspace $f(\omega _1,\omega _2)\le\mu (\omega _1)\le\mu (\omega _2)$. With the same terminology as in Case~1, we have
\begin{align*}
g_A(\omega _1)&=\int _Af(\omega ,\omega _1)d\phi (\omega )=f(\omega _1,\omega _2)a
  +\mu (\omega _1)-f(\omega _1,\omega _2)\\
g_A(\omega _2)&=\int _Af(\omega ,\omega _2)d\phi (\omega )=f(\omega _1,\omega _2)a
  +\mu (\omega _2)-f(\omega _1,\omega _2)
\end{align*}
Since $\mu$ is a measure, we have
\begin{align*}
\mu (\omega _1)+\mu (\omega _2)&=\mu (A)=\int _Ag_A(\omega ')d\phi (\omega ')\\
  &=\sqbrac{f(\omega _1,\omega _2)(a-1)+\mu (\omega _1)}a+\mu (\omega _2)-\mu (\omega _1)\\
  &=\mu (\omega _2)+(a-1)\mu (\omega _1)
\end{align*}
If $a=0$, then $\mu (\omega _1)+\mu (\omega _2)=\mu (\omega _2)-\mu (\omega _1)$ which is a contradiction. If $a=1$, then $\mu (\omega _1)+\mu (\omega _2)=\mu (\omega _2)$ which is a contradiction.

\medskip
\noindent\textbf{Case 3.}\enspace $\mu (\omega _1)\le f(\omega _1,\omega _2)\le\mu (\omega _2)$. Again, with the same notation as before, we have
\begin{align*}
g_A(\omega _1)&=\mu (\omega _1)a+f(\omega _1,\omega _2)-\mu (\omega _1)\\
g_A(\omega _2)&=f(\omega _1,\omega _2)a+\mu (\omega _2)-f(\omega _1,\omega _2)
\end{align*}
If $a=1$, since $\mu$ is a measure we have
\begin{align*}\mu (\omega _1)+\mu (\omega _2)&=\mu (A)=\int _Ag_A(\omega ')d\phi (\omega ')\\
  &=f(\omega _1,\omega _2)+\mu (\omega _2)-f(\omega _1,\omega _2)=\mu (\omega _2)
\end{align*}
which is a contradiction. If $a=0$, we have two subcases. If
\begin{equation*}
f(\omega _1,\omega _2)-\mu (\omega _1)\le\mu (\omega _2)-f(\omega _1,\omega _2)
\end{equation*}
then
\begin{equation*}
\mu (\omega _1)+\mu (\omega _2)=\mu (A)=\int _Ag_A(\omega ')d\phi (\omega ')
  =\mu (\omega _2)-f(\omega _1,\omega _2)
\end{equation*}
which is a contradiction. If
\begin{equation*}
\mu (\omega _2)-f(\omega _1,\omega _2)\le f(\omega _1,\omega _2)-\mu (\omega _1)
\end{equation*}
then
\begin{align*}
\mu (\omega _1)+\mu (\omega _2)&=\mu (A)=\int _Ag_A(\omega ')d\phi (\omega ')
  =f(\omega _1,\omega _2)-\mu (\omega _1)\\
  &\le\mu (\omega _2)-\mu (\omega _1)
\end{align*}
which is a contradiction.

Since Case~1 with $a=1$ is the only noncontradiction, we conclude that $\phi (\omega _i)=1$,
$f(\omega _i,\omega _i)=\mu (\omega _i)$, $i=1,2,3$, $\phi\paren{\brac{\omega _i,\omega _j}}=1$,
$f(\omega _i,\omega _j)=\mu (\omega _i)+\mu (\omega _j)$, $i<j=1,2,3$. Letting
$A=\brac{\omega _1,\omega _2,\omega _3}$ and $a=\phi\paren{\brac{\omega _1,\omega _2,\omega _3}}$ we obtain
\begin{align*}
g_A(\omega _1)&=\int _Af(\omega ,\omega _1)d\phi (\omega )=\mu (\omega _1)a+\mu (\omega _2)\\
g_A(\omega _2)&=\int _Af(\omega ,\omega _2)d\phi (\omega )=\mu (\omega _2)a+\mu (\omega _3)\\
g_A(\omega _3)&=\int _Af(\omega ,\omega _3)d\phi (\omega )=\mu (\omega _3)a+\mu (\omega _2)
\end{align*}
If $a=0$, since $\mu$ is a measure we have
\begin{equation*}
\mu (\omega _1)+\mu (\omega _2)+\mu (\omega _3)=\mu (A)=\int _Ag_A(\omega ')d\phi (\omega ')
  =\mu (\omega _3)-\mu (\omega _2)
\end{equation*}
which is a contradiction. If $a=1$, we obtain
\begin{equation*}
\mu (\omega _1)+\mu (\omega _2)+\mu (\omega _3)=\int _Ag_A(\omega ')d\phi (\omega ')
  =\mu (\omega _3)+\mu (\omega _2)
\end{equation*}
which is again a contradiction. Since every case leads to a contradiction, $\mu$ is not 2-generating.
\end{proof}

\section{1-Generation} 
This section mainly considers $q$-measures and their 1-generated coevents in $\ascript _2$ and $\ascript _3$. We begin by showing that only a very restricted set of $q$-measures 1-generate coevents.

\begin{lem}       
\label{lem61}
Let $\mu$ be a $q$-measure on $\ascript$ that 1-generates a coevent $\phi\in\ascript ^*$. If
$\omega _1,\omega _2\in\Omega$ with $0<\mu (\omega _1)\le\mu (\omega _2)$, then
$\mu\paren{\brac{\omega _1,\omega _2}}=\mu (\omega _2)-\mu (\omega _1)$ or
$\mu\paren{\brac{\omega _1,\omega _2}}=\mu (\omega _2)$. Moreover,
$\mu\paren{\brac{\omega _1,\omega _2}}=\mu (\omega _2)-\mu (\omega _1)$ if and only if
$\phi\paren{\brac{\omega _1,\omega _2}}=0$ and $\mu\paren{\brac{\omega _1,\omega _2}}=\mu (\omega _2)$ if and only if $\phi\paren{\brac{\omega _1,\omega _2}}=1$.
\end{lem}
\begin{proof}
Let $f$ be a $\phi$-density for $\mu$. Since $0<\mu (\omega _1)\le\mu (\omega _2)$, by Lemma~\ref{lem55} we have that
\begin{equation*}
f(\omega _1)=\mu (\omega _1)\le\mu (\omega _2)=f(\omega _2)
\end{equation*}
and $\phi (\omega _2)=1$. Hence,
\begin{align*}
\mu\paren{\brac{\omega _1,\omega _2}}&=\int _{\brac{\omega _1,\omega _2}}fd\phi
  =f(\omega _1)\phi\paren{\brac{\omega _1,\omega _2}}
  +\sqbrac{f(\omega _2)-f(\omega _1)}\phi\paren{\brac{\omega _2}}\\
  &=\mu (\omega _1)\phi\paren{\brac{\omega _1,\omega _2}}+\mu (\omega _2)-\mu (\omega _1)
\end{align*}
It follows that $\mu\paren{\brac{\omega _1,\omega _2}}=\mu (\omega _2)-\mu (\omega _1)$ if and only if
$\phi\paren{\brac{\omega _1\omega _2}}=0$ and $\mu\paren{\brac{\omega _1,\omega _2}}=\mu (\omega _2)$
if and only if $\phi\paren{\brac{\omega _1,\omega _2}}=1$.
\end{proof}

\begin{exam}{7}    
This example considers $q$-measures and 1-generated coevents on $\Omega _2=\brac{\omega _1,\omega _2}$. The zero measure 1-generates $0\in\ascript _2^*$ and Dirac measures $c\delta _{\omega _i}$ 1-generate
$\omega _i^*$, $i=1,2$. For the type $(1,2)$ coevent $\phi=\omega _1^*\oplus\omega _2^*$, let
$f\colon\Omega _2\to\real$ be strictly positive and define $\mu (A)=\int _Afd\phi$, $A\in\ascript _2$. Assuming that
$f(\omega _1)\le f(\omega _2)$ we have that $\mu (\omega _1)=f(\omega _1)$, $\mu (\omega _2)=f(\omega _2)$ and $\mu (\Omega _2)=f(\omega _2)-f(\omega _1)$. Hence, any $q$-measure on $\ascript _2$ that satisfies
$0<\mu (\omega _1)\le\mu (\omega _2)$, $\mu (\Omega _2)=\mu (\omega _2)-\mu (\omega _1)$ 1-generates the coevent $\omega _1^*\oplus\omega _2^*$. The density is given by $f(\omega _1)=\mu (\omega _1)$,
$f(\omega _2)=\mu (\omega _2)$ and is unique.

For the type $(12)$ coevent $\phi =\omega _1^*\omega _2^*$, let $f\colon\Omega _2\to\real$,
$\mu (A)=\int _Afd\phi$ as before. Assuming $f(\omega _1)\le f(\omega _2)$, we have that
$\mu (\omega _1)=\mu (\omega _2)=0$ and $\mu (\Omega _2)=f(\omega _1)$. Notice that $\mu$ is not regular. Hence, any $q$-measure $\mu$ on $\ascript _2$ that satisfies $\mu (\omega _1)=\mu (\omega _2)=0$ 1-generates the coevent $\omega _1^*\omega _2^*$. The density is given by $f(\omega _2)\ge f(\omega _1)=\mu (\Omega _2)$ but otherwise is arbitrary. In this case, the density is not unique.

For the type $(1,12)$ coevent $\phi =\omega _1^*\oplus\omega _1^*\omega _2^*$, let
$f\colon\Omega _2\to\real$, $\mu (A)=\int _Ad\phi$ as before. Assuming $f(\omega _2)\le f(\omega _1)$ we have that
$\mu (\omega _1)=f(\omega _1)$, $\mu (\omega _2)=0$ and $\mu (\Omega _2)=\mu (\omega _1)-f(\omega _2)$. Notice that $\mu$ is not regular. Hence, any $q$-measure on $\ascript _2$ that satisfies $\mu (\omega _1)>0$,
$\mu (\omega _2)=0$, $\mu (\Omega _2)\le\mu (\omega _1)$ 1-generates the coevent
$\omega _1^*\oplus\omega _1^*\omega _2^*$. The density is given by $f(\omega _1)=\mu (\omega _1)$,
$f(\omega _2)=\mu (\omega _1)-\mu (\Omega _2)$ and is unique.

For the type $(1,2,12)$ coevent 
\begin{equation*}
\Omega _2^*=1=\omega _1^*\oplus\omega _2^*\oplus\omega _1^*\omega _2^*
\end{equation*}
let $f\colon\Omega _2\to\real$, $\mu (A)=\int _Afd\phi$ as before. Assuming $f(\omega _1)\le f(\omega _2)$ we have that $\mu (\omega _1)=f(\omega _1)$, $\mu (\omega _2)=f(\omega _2)$ and $\mu (\Omega _2)=f(\omega _2)$. Hence, any $q$-measure $\mu$ on $\ascript _2$ that satisfies
\begin{equation*}
0<\mu (\omega _1)\le\mu (\omega _2)=\mu (\Omega _2)
\end{equation*}
1-generates the coevent $1$ and the unique density is
$f(\omega _1)=\mu (\omega _1)$, $f(\omega _2)=\mu (\omega _2)$.\hfill\qedsymbol
\end{exam}
\medskip

Examining all the cases in Example~7 shows that every coevent $\phi\in\ascript _2^*$ is 1-generated and when a
$q$-measure on $\ascript _2$ 1-generates a coevent $\phi$, then $\phi$ is unique.

\begin{exam}{8}    
There are too many coevents in $\ascript _3^*$ to consider them all so we give some examples in $\ascript _3^*$ that are 1-generated and some that are not 1-generated. It follows from Theorem~\ref{thm54} that
$\omega _i^*$, $i=1,2,3$, $\omega _i^*\oplus\omega _j^*$ and
$\omega _i^*\oplus\omega _j^*\oplus\omega _i^*\omega _j^*$, $i<j=1,2,3$, are 1-generated. We now describe a set of 1-generated coevents in $\ascript _3^*$ that include these nine coevents. Since $\phi (A)=\int _Ad\phi$, if
$\phi\colon\ascript\to\brac{0,1}$ happens to be $q$-measure, then $\phi$ 1-generates itself with density $f=1$. In general, $\phi$ may not be regular. We now describe the 34 $q$-measures in $\ascript _3^*$. List the nonempty subsets of $\Omega _3$ in the order $A_i$, $i=1,\ldots ,7$ as follows
\begin{equation*}
\brac{\omega _1},\brac{\omega _2},\brac{\omega _3},\brac{\omega _1,\omega _2},\brac{\omega _1,\omega _3}
\brac{\omega _2,\omega _3},\Omega _3
\end{equation*}
We can represent a coevent by $\phi _{a_1\cdots a_7}$ where $a_i\in\brac{0,1}$ are not all zero and $a_i=1$ if and only if $\phi (A_i)=1$. Now $\phi _{a_1\cdots a_7}$ is a $q$-measure if and only if
$a_7=a_4+a_5+a_6-a_1-a_2-a_3$. If we count these according to the number of ones in $\brac{a_1,a_2,a_3}$ we obtain $3\cdot 4+3\cdot 6+4=34$ $q$-measures. Examples are
\begin{align*}
&\omega _1^*=\phi _{1001101},\quad\omega _1^*\oplus\omega _2^*=\phi _{1100110},\quad
(\omega _1^*\omega _2^*\omega _3^*)'=\phi _{1111110}\\
&\omega _2^*\oplus\omega _1^*\omega _2^*\oplus\omega _1^*\omega _3^*=\phi _{0100111}
\end{align*}
For instance, $\phi =(\omega _1^*\omega _2^*\omega _3^*)'$ is 1-generated by itself. Moreover, $\phi$ is
1-generated by any $q$-measure $\mu$ satisfying $\mu (\omega _i)>0$, $i=1,2,3$,
$\mu\paren{\brac{\omega _i,\omega _j}}=\max\paren{\mu (\omega _i),\mu (\omega _j)}$, $i<j=1,2,3$. The density is
$f(\omega _i)=\mu (\omega _i)$.

We now give examples of coevents in $\ascript _3^*$ that are not 1-generated. First
$\phi =\omega _1^*\omega _2^*\omega _3^*$ is not 1-generated. Suppose $\mu$ is a $q$-measure on 
$\ascript _3$ and $\mu (A)=\int _Afd\phi$ where $f\colon\Omega _3\to\real$ is strictly positive. We can assume without loss of generality that $0<f(\omega _1)\le f(\omega _2)\le f(\omega _3)$. By Lemma~\ref{lem55} we have that $\mu (\omega _i)=0$, $i=1,2,3$. Moreover, for $i<j=1,2,3$ we obtain
\begin{equation*}
\mu\paren{\brac{\omega _i,\omega _j}}=\int _{\brac{\omega _i,\omega _j}}fd\phi =0
\end{equation*}
and
\begin{equation*}
\mu (\Omega _3)=\int _{\Omega _3}fd\phi =f(\omega _1)>0
\end{equation*}
Since $\mu$ is a $q$-measure we conclude that
\begin{align*}
\mu (\Omega _3)&=\mu\paren{\brac{\omega _1,\omega _2}}+\mu\paren{\brac{\omega _1,\omega _3}}
  +\mu\paren{\brac{\omega _2,\omega _3}}-\mu (\omega _1)-\mu (\omega _2)-\mu (\omega _3)\\
  &=0
\end{align*}
which is a contradiction.

We next show that $\phi =\omega _1^*\oplus\omega _2^*\oplus\omega _3^*$ is not 1-generated. Suppose $\mu$ is a $q$-measure on $\ascript _3$ and $\mu (A)=\int _Afd\phi$ where $f\colon\Omega _3\to\real$ is strictly positive. We can assume that $0<f(\omega _1)\le f(\omega _2)\le f(\omega _3)$. By Lemma~\ref{lem55} we have
$f(\omega _i)=\mu (\omega _i)$, $i=1,2,3$. Now
\begin{equation*}
\mu\paren{\brac{\omega _1,\omega _2}}=\int _{\brac{\omega _1,\omega _2}}fd\phi
  =f(\omega _2)-f(\omega _1)=\mu (\omega _2)-\mu (\omega _1)
\end{equation*}
and similarly, $\mu\paren{\brac{\omega _1,\omega _3}}=\mu (\omega _3)-\mu (\omega _1)$,
$\mu\paren{\brac{\omega _2,\omega _3}}=\mu (\omega _3)-\mu (\omega _2)$. We also have that
\begin{equation*}
\mu (\Omega _3)=\int _{\Omega _3}fd\phi =f(\omega _1)+f(\omega _3)-f(\omega _2)
  =\mu (\omega _3)-\mu (\omega _2)+\mu (\omega _1)
\end{equation*}
Since $\mu$ is a $q$-measure, we obtain
\begin{align*}
\mu (\omega _3)&-\mu (\omega _2)+\mu (\omega _1)=\mu (\Omega _3)\\
  &=\mu\paren{\brac{\omega _1,\omega _2}}+\mu\paren{\brac{\omega _1,\omega _3}}
  +\mu\paren{\brac{\omega _2,\omega _3}}-\mu (\omega _1)-\mu (\omega _2)-\mu (\omega _3)\\
  &=\mu (\omega _3)-\mu (\omega _2)-3\mu (\omega _1)
\end{align*}
Since this gives a contradiction, $\phi$ is not 1-generated.

Finally, we show that $\phi =\omega _1^*\oplus\omega _2^*\oplus\omega _3^*\oplus\omega _1^*\omega _2^*$ is not 1-generated. Suppose $\mu$ is a $q$-measure $\ascript _3$ and $\mu (A)=\int _Afd\mu$ where
$f\colon\Omega _3\to\real$ is strictly positive. By Lemma~\ref{lem55}, $f(\omega _i)=\mu (\omega _i)$, $i=1,2,3$. We now have three cases.\newline
\textbf{Case 1.}\enspace $0<f(\omega _1)\le f(\omega _2)\le f(\omega _3)$. We obtain
\begin{align*}
\mu\paren{\brac{\omega _1,\omega _2}}&=\int _{\brac{\omega _1,\omega _2}}fd\phi =\mu (\omega _2)\\
\mu\paren{\brac{\omega _1,\omega _3}}&=\int _{\brac{\omega _1,\omega _3}}fd\phi
  =\mu (\omega _3)-\mu (\omega _1)\\
\mu\paren{\brac{\omega _2,\omega _3}}&=\int _{\brac{\omega _2,\omega _3}}fd\phi
  =\mu (\omega _3)-\mu (\omega _2)\\
  \mu (\Omega _3)&=\int _{\Omega _3}fd\phi =\mu (\omega _3)-\mu (\omega _2)
\end{align*}
Since $\mu$ is a $q$-measure, we have
\begin{equation*}
\mu (\omega _3)-\mu (\omega _2)=\mu (\Omega _3)=\mu (\omega _3)-\mu (\omega _2)-2\mu (\omega _1)
\end{equation*}
which is a contradiction.\newline
\textbf{Case 2.}\enspace $0<f(\omega _3)\le f(\omega _1)\le f(\omega _2)$. We obtain
\begin{align*}
\mu\paren{\brac{\omega _1,\omega _2}}&=\int _{\brac{\omega _1,\omega _2}}fd\phi =\mu (\omega _2)\\
\mu\paren{\brac{\omega _1,\omega _3}}&=\int _{\brac{\omega _2,\omega _3}}fd\phi
  =\mu (\omega _1)-\mu (\omega _3)\\
\mu\paren{\brac{\omega _2,\omega _3}}&=\int _{\brac{\omega _2,\omega _3}}fd\phi
  =\mu (\omega _2)-\mu (\omega _3)\\
\mu (\Omega _3)&=\int _{\Omega _3}fd\phi =\mu (\omega _2)-\mu (\omega _3)
\end{align*}
Since $\mu$ is a $q$-measure, we have
\begin{equation*}
\mu (\omega _2)-\mu (\omega _3)=\mu (\Omega _3)=\mu (\omega _2)-3\mu (\omega _3)
\end{equation*}
which is a contradiction.\newline
\textbf{Case 3.}\enspace 
$0<f(\omega _1)\le f(\omega _3)\le f(\omega _2)$. In a similar way as before, we obtain
$\mu\paren{\brac{\omega _1,\omega _2}}=\mu (\omega _2)$,
$\mu\paren{\brac{\omega _1,\omega _3}}=\mu (\omega _3)-\mu (\omega _1)$,
$\mu\paren{\brac{\omega _2,\omega _3}}=\mu (\omega _2)-\mu (\omega _3)$ and
$\mu (\Omega _3)=\mu (\omega _2)-\mu (\omega _3)$. Since $\mu$ is a $q$-measure, we have
\begin{equation*}
\mu (\omega _2)-\mu (\omega _3)=\mu (\Omega _3)=\mu (\omega _2)-2\mu (\omega _1)-\mu (\omega _3)
\end{equation*}
which is a contradiction.\hfill\qedsymbol
\end{exam}

\section{2-Generation} 
This section illustrates by examples that more $q$-measures are 2-generating than 1-generating and more coevents are 2-generated than are 1-generated.

\begin{exam}{9}    
We have seen in Example~7 that only $q$-measures on $\ascript _2$ that satisfy
$\mu (\omega _1),\mu (\omega _2)>0$ and
\begin{equation*}
\mu (\Omega _2)=\max\paren{\mu (\omega _1),\mu (\omega _2)}-\min\paren{\mu (\omega _1),\mu (\omega _2)}
\end{equation*}
1-generate $\phi=\omega _1^*\oplus\omega _2^*$. We now show that if a $q$-measure $\mu$ on $\ascript _2$ satisfies $\mu (\omega _1),\mu (\omega _2)>0$ and
\begin{equation*}
\mu (\Omega _2)\le\max\paren{\mu (\omega _1),\mu (\omega _2)}-\min\paren{\mu (\omega _1),\mu (\omega _2)}
\end{equation*}
then $\mu$ 2-generates $\phi=\omega _1^*\oplus\omega _2^*$. First assume without loss of generality that
$0<\mu (\omega _1)\le\mu (\omega _2)$ and $\mu (\Omega _2)\le\mu (\omega _2)-\mu (\omega _1)$. Let
$f\colon\Omega _2\times\Omega _2\to\real$ be the strictly positive, symmetric function defined by
$f(\omega _i,\omega _i)=\mu (\omega _i)$, $i=1,2$ and
\begin{equation*}
f(\omega _1,\omega _2)=f(\omega _2,\omega _1)=\frac{\mu (\Omega _2)+\mu (\omega _1)+\mu (\omega _2)}{2}
\end{equation*}
We then have
\begin{equation*}
\mu (\omega _1)\le\frac{\mu (\omega _1)+\mu (\omega _2)}{2}
  \le\frac{\mu (\Omega _2)+\mu (\omega _1)+\mu (\omega _2)}{2}\le\mu (\omega _2)
\end{equation*}
Hence, $f(\omega _1,\omega _1)\le f(\omega _1,\omega _2)\le f(\omega _2,\omega _2)$. Define the set function
$\mu '\colon\ascript _2\to\real ^+$ by
\begin{equation*}
\mu '(A)=\int _A\sqbrac{\int _Af(\omega ,\omega ')d\phi (\omega )}d\phi (\omega ')
\end{equation*}
Then
\begin{equation*}
\mu '(\omega _1)=\int _{\brac{\omega _1}}f(\omega _1,\omega ')d\phi (\omega ')=f(\omega _1,\omega _1)
  =\mu (\omega _1)
\end{equation*}
and similarly, $\mu '(\omega _2)=\mu (\omega _2)$. Defining
$g(\omega ')=\int f(\omega ,\omega ')d\phi (\omega )$ we have
\begin{align*}
g(\omega _1)&=\int f(\omega ,\omega _1)d\phi (\omega )=f(\omega _1,\omega _2)-f(\omega _1,\omega _1)
  =f(\omega _1,\omega _2)-\mu (\omega _1)\\
g(\omega _2)&=\int f(\omega ,\omega _2)d\phi (\omega )=f (\omega _2,\omega _2)-f(\omega _1,\omega _2)
  =\mu (\omega _2)-f(\omega _1,\omega _2)
\end{align*}
Since
\begin{equation*}
f(\omega _1,\omega _2)-\mu (\omega _1)-\mu (\omega _2)+f(\omega _1,\omega _2)
  =2f(\omega _1,\omega _2)-\mu (\omega _1)-\mu (\omega _2)=\mu (\Omega _2)\ge 0
\end{equation*}
we have that
\begin{equation*}
\mu '(\Omega _2)=\int g(\omega ')d\phi (\omega ')=\mu (\Omega _2)
\end{equation*}
Hence, $\mu (A)=\mu '(A)$ for all $A\in\ascript _2$ so $\mu$ 2-generates $\phi$.\hfill\qedsymbol
\end{exam}

Just as every coevent in $\ascript _2^*$ is 1-generated, it is not hard to show that every coevent in $\ascript _2^*$ is 2-generated. For the same reason that $\omega _1^*\omega _2^*\omega _3^*$ is not 1-generated in $\ascript _3$ we have that $\omega _1^*\omega _2^*\omega _3^*$ is not 2-generated in $\ascript _3$. In fact, for $m\ge 3$,
$\phi =\omega _1^*\cdots\omega _m^*$ is not $1$ or 2-generated in $\ascript _n$, $n\ge m$. This is because
$\phi (A)=0$ for all $A\in\ascript _n$ such that $\brac{\omega _1,\ldots ,\omega _m}\not\subseteq A$. Hence, if a
$q$-measure $\mu$ $1$ or 2-generates $\phi$, then
\begin{equation*}
\mu (\omega _i)=\mu\paren{\brac{\omega _i,\omega _j}}=0
\end{equation*}
for all $i,j\le n$. However, $\mu\paren{\brac{\omega _1,\ldots ,\omega _m}}>0$ and this contradicts \eqref{eq42}. The same reasoning shows that any coevent whose evaluation map representation has all terms of degree larger than $2$ is not $1$ or 2-generated.

\begin{exam}{10}    
We have seen in Example~8, that
\begin{equation*}
\phi =\omega _1^*\oplus\omega _2^*\oplus\omega _3^*\oplus\omega _1^*\omega _2^*
\end{equation*}
is not 1-generated in $\ascript _3$. We now show that $\phi$ is 2-generated in $\ascript _3$. Place a $q$-measure
$\mu$ on $\ascript _3$ satisfying $0<\mu (\omega _1)\le\mu (\omega _2)$,
$\mu (\omega _3)=\mu (\omega _1)+\mu (\omega _2)$,
$\mu\paren{\brac{\omega _1,\omega _2}}\ge\mu (\omega _3)$,
$\mu\paren{\brac{\omega _1,\omega _3}}=\mu (\omega _2)$,
$\mu\paren{\brac{\omega _2,\omega _3}}=\mu (\omega _1)$ and
$\mu (\Omega _3)=\mu\paren{\brac{\omega _1,\omega _2}}-\mu (\omega _3)$. To show that $\mu$ is indeed a
$q$-measure we have
\begin{equation*}
\sum _{i<j=1}^3\mu\paren{\brac{\omega _i,\omega _j}}-\sum _{i=1}^3\mu (\omega _i)
  =\mu\paren{\brac{\omega _1,\omega _2}}-\mu (\omega _3)=\mu (\Omega _3)
\end{equation*}
Let $f\colon\Omega _3\times\Omega _3\to\real$ be the strictly positive, symmetric function satisfying
$f(\omega _i,\omega _i)=\mu (\omega _i)$, $i=1,2,3$ and
\begin{equation*}
f(\omega _1,\omega _2)=f(\omega _1,\omega _3)=f(\omega _2,\omega _3)=\mu\paren{\brac{\omega _1,\omega _2}}
\end{equation*}
Letting $\mu '(A)=\int _A\sqbrac{\int _Af(\omega ,\omega ')d\phi (\omega )}d\phi (\omega ')$ for all $A\in\ascript _3$ we have $\mu '(\omega _i)=\mu (\omega _i)$, $i=1,2,3$. Moreover,
\begin{equation*}
g_{\brac{\omega _1,\omega _2}}(\omega _1)
  =\int _{\brac{\omega _1,\omega _2}}f(\omega ,\omega _1)d\phi (\omega )
  =f(\omega _1,\omega _2)=\mu\paren{\brac{\omega _1,\omega _2}}
\end{equation*}
Similarly, $g_{\brac{\omega _1,\omega _2}}(\omega _2)=\mu\paren{\brac{\omega _1,\omega _2}}$ and for the other doubleton sets we have
\begin{align*}
g_{\brac{\omega _1,\omega _3}}(\omega _1)&=\mu\paren{\brac{\omega _1,\omega _2}}-\mu (\omega _1)\\
g_{\brac{\omega _1,\omega _3}}(\omega _3)&=\mu\paren{\brac{\omega _1,\omega _2}}-\mu (\omega _3)\\
g_{\brac{\omega _2,\omega _3}}(\omega _2)&=\mu\paren{\brac{\omega _1,\omega _2}}-\mu (\omega _2)\\
g_{\brac{\omega _2,\omega _3}}(\omega _3)&=\mu\paren{\brac{\omega _1,\omega _2}}-\mu (\omega _3)\\
\end{align*}
Hence,
\begin{equation*}
\mu '\paren{\brac{\omega _1,\omega _2}}
  =\int _{\brac{\omega _1,\omega _2}}g_{\brac{\omega _1,\omega _2}}(\omega ')d\phi (\omega ')
  =\mu\paren{\brac{\omega _1,\omega _2}}
\end{equation*}
and similarly, $\mu '\paren{\brac{\omega _1,\omega _3}}=\mu\paren{\brac{\omega _1,\omega _3}}$,
$\mu '\paren{\brac{\omega _2,\omega _3}}=\mu\paren{\brac{\omega _2,\omega _3}}$.\break Finally,
\begin{equation*}
g_\Omega (\omega _1)=\int f(\omega ,\omega _1)d\phi (\omega )
  =f(\omega _1,\omega _3)-f(\omega _1,\omega _2)=0
\end{equation*}
and similarly, $g_\Omega (\omega _2)=0$, $g_\Omega (\omega _3)=\mu (\Omega _3)$. We conclude that
\begin{equation*}
\mu '(\Omega _3)=\int g_\Omega (\omega ')d\phi (\omega ')=\mu (\Omega _3)
\end{equation*}
Hence, $\mu (A)=\mu '(A)$ for all $A\in\ascript _3$ so $\mu$ 2-generates $\phi$.\hfill\qedsymbol
\end{exam}
\medskip

We do not know whether $\omega _1^*\oplus\omega _2^*\oplus\omega _3^*$ is 2-generated in $\ascript _3$.

\end{document}